\documentclass[12pt]{article}
\usepackage{times}

\usepackage{amsmath,amsthm,amssymb,bbold, mathtools,setspace,enumitem,epsfig,titlesec,verbatim,color,array,multirow,comment,graphicx,hyperref,blkarray}
\usepackage[sort&compress,comma,round,numbers]{natbib} 
\usepackage[bf,small]{caption}
\usepackage[export]{adjustbox}
\usepackage[margin=2.5cm, includefoot, footskip=30pt]{geometry}

\makeatletter
\renewcommand{\fnum@figure}{Fig. \thefigure}
\makeatother

\definecolor{darkred}{rgb}{0.6,0,0}
\definecolor{darkblue}{rgb}{0,0.3,0.8}

\titleformat{\section}{\sffamily \fontsize{13}{20}\bfseries}{\thesection}{1em}{}
\titleformat{\subsection}{\sffamily \fontsize{12}{20}\bfseries}{\thesubsection}{1em}{}

\newtheoremstyle{plainCl1}
{9pt}
{15pt}
{\it}
{}
{\bfseries}
{.}
{2mm}
{}

\newtheoremstyle{plainCl2}
{9pt}
{15pt}
{\it}
{}
{\bfseries}
{.}
{4mm}
{}

\theoremstyle{plainCl1}

\theoremstyle{plainCl2}

\newcommand{\SupplementaryInformation}{\textbf{Supplementary Information}}
\newcommand{\Methods}{\textbf{Materials and Methods}}

\title{\sffamily \large {\bfseries The co-evolution of direct, indirect and generalized reciprocity}}
\date{\empty}
\author{\parbox[c]{16cm}{\centering \onehalfspacing \fontsize{11}{12}\selectfont Saptarshi Pal$^{1,2}$, Christian Hilbe$^{2}$, Nikoleta E. Glynatsi$^{2}$ \\[0.2cm]
$^1$ Department of Mathematics, Harvard University, Cambridge, USA.\\
$^2$Max Planck Research Group Dynamics of Social Behavior,\\ Max Planck Institute for Evolutionary Biology, Germany.}\\ \\
}

\begin{document}
\maketitle
\onehalfspacing

\section*{Abstract}
People often engage in costly cooperation, especially in repeated interactions.
When deciding whether to cooperate, individuals typically take into account how
others have acted in the past. For instance, when one person is deciding whether
to cooperate with another, they may consider how they were treated by the other
party (direct reciprocity), how the other party treated others (indirect
reciprocity), or how they themselves were treated by others in general
(generalized reciprocity). Given these different approaches, it is unclear which
strategy, or more specifically which mode of reciprocity, individuals will
prefer. This study introduces a model where individuals decide how much weight
to give each type of information when choosing to cooperate. Through equilibrium
analysis, we find that all three modes of reciprocity can be sustained when
individuals have sufficiently frequent interactions. However, the existence of
such equilibria does not guarantee that individuals will learn to use them.
Simulations show that when individuals mainly imitate others, generalized
reciprocity often hinders cooperation, leading to defection even under
conditions favorable to cooperation. In contrast, when individuals explore new
strategies during learning, stable cooperation emerges through direct
reciprocity. This study highlights the importance of studying all forms of
reciprocity in unison.
\newpage
\section*{INTRODUCTION}

Reciprocity is one of the primary reasons why people cooperate with
others~\citep{Nowak:Science:2006}. The promise of a reciprocal act often drives
us to help others, even when it is costly to do so. However, the concept of
reciprocity extends beyond direct reciprocity, the simple give and take between
two individuals. People can also reciprocate by helping those who help others,
even without benefiting
directly~\citep{Wedekind:Science:2000,Milinski:PRSB:2001,Seinen:EER:2006,Engelmann:GEB:2009,Yoeli:PNAS:2013},
or by extending help to others after having
received help
themselves~\citep{baker2014paying,melamed2020robustness,simpson2018roots}.
These alternative forms of repayment are referred to as downstream and upstream
reciprocity, or simply as indirect~\citep{Nowak:Nature:2005,Ohtsuki:JTB:2006a,
Ohtsuki:JTB:2004a,clark2020indirect} and generalized
reciprocity~\citep{Pfeiffer:PRSB:2005,barta2011cooperation,
Rutte:PLoSB:2007,Nowak:Nature:2005}.

Although the three kinds of reciprocity—direct, indirect, and generalized—are
conceptually different, they are intertwined in practice. For example, if Alice
has a positive experience with Bob but also notices that he defects with others,
should she employ direct reciprocity (to cooperate) or indirect reciprocity (to
defect) with him the next time they meet? Similarly, if Bob observes Alice
defecting with others and is therefore unwilling to cooperate with her (as per
indirect reciprocity), should he change his mind if Charlie just cooperated with
him and he is in a positive mood (as per generalized reciprocity)? Most previous
studies are unable to formally test these types of questions, either because
they only study one kind of
reciprocity~\citep{Axelrod:Science:1981,Ohtsuki:JTB:2004a,Pfeiffer:PRSB:2005,van2012evolution,fudenberg2012slow,Van-Veelen:PNAS:2012,Yoeli:PNAS:2013,Hilbe:ncomms:2014,Hilbe:NHB:2018},
or each kind is studied in
isolation~\citep{Boyd:SN:1989,dufwenberg2001direct,herne2013experimental,stanca2009measuring}.

By now, experimental~\citep{Molleman:PRSB:2016} and theoretical
papers~\citep{Schmid:NHB:2021,nakamaru2004evolution,seki2016model} have studied
how individuals weigh between direct and indirect reciprocity. The broad
conclusion from these studies is that people often integrate their own direct
experience and the indirect experience of others to decide whether to cooperate
with someone~\citep{Molleman:PRSB:2016}. In addition, there are situations where
individuals predominantly prefer direct over indirect reciprocity. For example,
when outside information is untrustworthy, individuals tend to ignore it and
rely on their direct
experience~\citep{Schmid:NHB:2021,seki2016model,nakamaru2004evolution}. These
results on direct and indirect reciprocity lead up to the natural follow-up
question: how do we integrate our general experience with others when we decide
to cooperate with someone? In particular, are there situations in which
individuals may prefer to cooperate via generalized reciprocity over the two
other modes? Most theoretical models of generalized reciprocity either treat it
as an isolated form of
reciprocity~\citep{rankin2009assortment,Pfeiffer:PRSB:2005,barta2011cooperation,van2012evolution},
or are, by construction, unable to study the competition between the different
modes of reciprocity~\citep{nowak2007upstream,sasaki2023evolution}. To address
the questions posed above, however, a unified model is crucial.

In this paper, we introduce a model in which players can adopt strategies that
cooperate through direct, indirect, or generalized reciprocity, or a combination
of these mechanisms. More precisely, while interacting repeatedly within a
population, individuals cooperate based on whether they perceive their current
co-player to be in good or bad standing. These standings are updated based on
strategies that assign different weights to direct, indirect, and generalized
experiences. Our model is a direct extension of the framework introduced by
Schmid et al.~\citep{Schmid:NHB:2021}, where they consider direct and indirect
reciprocity in a unified model. Schmid et al. were the first to propose a
framework that is explicitly analyzable and allows for analytical conclusions
when comparing direct and indirect reciprocity. Similarly, our model remains
mathematically tractable. We also extend the theory of zero-determinant
strategies from direct reciprocity to a framework that includes indirect and
generalized reciprocity, enabling us to derive analytical conclusions when
comparing the three mechanisms.

\section*{RESULTS}

\subsection*{A model of direct, indirect and generalized reciprocity}

We consider a population of \(n\) individuals who repeatedly interact in pairs.
More specifically, each interaction involves two
randomly selected players from the population, who simultaneously and
independently decide whether to cooperate or defect. If both cooperate, they
incur a cost \(c > 0\) to provide a benefit \(b > c\) to their co-player. If
they defect, they pay no cost and provide no benefit. All individuals not
participating in the current interaction observe its outcome. However, these
third-party observations are subject to errors: each observer independently
mistakes cooperation for defection (or \emph{vice versa}) with probability \(
\varepsilon\). We assume that players do not make observation errors in
interactions they are directly involved in. Once the interaction concludes,
another one occurs with probability \(d\) and with probability \(1 - d\),
the game ends. In the next interaction, two individuals are again randomly
paired to interact in the same manner as before. The game's setup is illustrated
in Fig.  \ref{fig:introductory}A. At the end of the game, players' payoffs are
calculated by averaging over all the interactions in which they participated.

The decision to cooperate or defect in an interaction depends on how players
currently perceive their co-player. Players can view their co-players in one of
two binary standings: ``good'' or ``bad''. If a player perceives their co-player as
good, they cooperate; otherwise, they defect (Fig. \ref{fig:introductory}B).
These interpersonal standings are
referred to as reputations. Reputations are updated throughout the game based on
the strategies of individuals (Fig. \ref{fig:introductory}C). A player's
strategy, denoted as \(\sigma = (y, p, q, \lambda, \gamma)\), is a tuple with
five components. The parameter \(y\) represents the initial probability that a
co-player is deemed good in the absence of prior information. This applies when
the player is interacting with someone for the first time and has yet to form an
opinion. The strategy components \(p\) and \(q\) determine the likelihood of
assigning a good reputation to a co-player. 
Specifically, \(p\) represents the probability of assigning a good reputation
after cooperation, while \(q\) represents the probability after defection.

Besides  direct interactions,  players may also update reputations based on
indirect information. If a co-player interacts with a third party, the focal
player updates that co-player's reputation with probability \(\lambda\)
(indirect reciprocity). In this case, the co-player is assigned a good
reputation with probability \(p\) (or \(q\)), if the co-player cooperated (or defected).
In addition, after a direct interaction with a third party,  the focal
player herself updates the reputation of every uninvolved co-player, independently, 
with probability \(\gamma\) (generalized reciprocity).
Once again, reputations are updated according to \(p\) or \(q\) based on whether 
the focal player's interaction partner cooperated or defected.
For a more detailed description of the strategies, see Section 1 of the
\SupplementaryInformation.

\begin{figure}[t!]
    \centering
    \includegraphics[width = \textwidth]{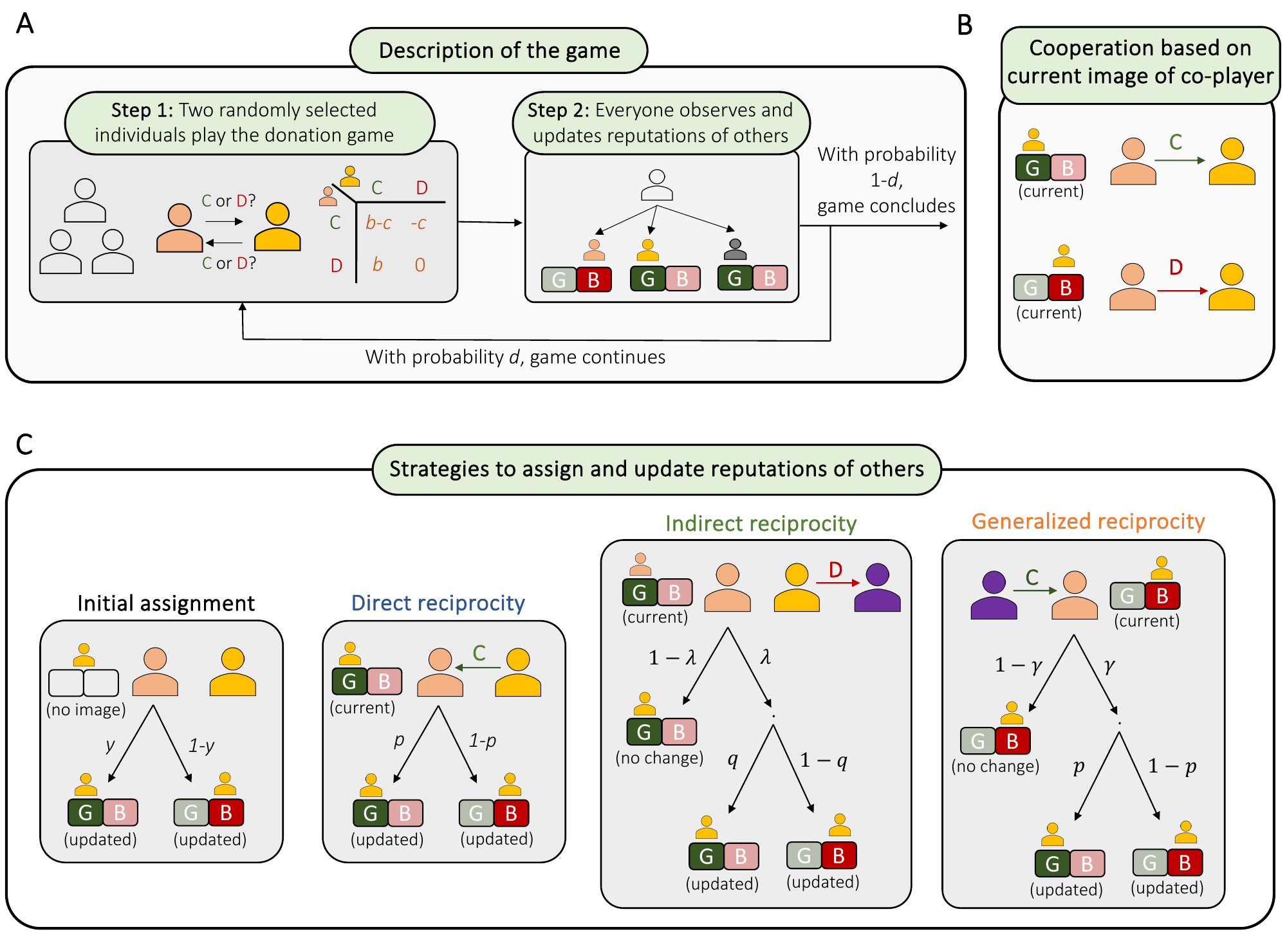}
    \caption{\textbf{A model with direct, indirect and generalized reciprocity.}
    \textbf{(A)} Individuals in a population of size \(n\) interact repeatedly in
    pairs, where each interaction involves two randomly selected players who
    independently choose whether to cooperate or defect. Cooperation results in
    a cost \(c\) and a benefit \(b > c\) to the co-player, while defection
    incurs no cost or benefit. Third-party observers, not involved in the
    interaction, observe the outcome, though these observations may be erroneous
    with probability \(\varepsilon\). Interactions continue with probability
    \(d\), and the game ends with probability \(1 - d\). 
    \textbf{(B)} Players assign reputations (good or bad) to their co-players,
    and their decision to cooperate or defect depends on the current perception
    of their co-player's standing.
    \textbf{(C)} Players assign a good reputation to co-players for whom there is no existing assessment
    with a probability \(y\).
    Reputations are updated based on both direct interactions and
    indirect information, with parameters \(p\) and \(q\) governing updates
    after cooperation or defection, respectively. Indirect
    reciprocity allows reputation updates from third-party interactions with
    probability \(\lambda\), and generalized reciprocity allows spillover
    updates for uninvolved individuals with probability \(\gamma\).}
    \label{fig:introductory}
\end{figure}

Let us now consider the extreme cases of the model. When \(\lambda = \gamma =
0\), players rely exclusively on direct reciprocity, meaning reputations are
updated solely based on their direct interactions with co-players. In this case,
reputations change only according to the outcomes of a player's personal
experiences. If \(\lambda = 0\) and \(\gamma = 1\), reputations are still
updated based on the player's own interactions, but the update does not involve
the co-player directly. Instead, players use their most recent interaction to
update the reputation of another individual. This scenario represents pure
generalized reciprocity. With \(\gamma = 0\) and \(\lambda = 1\),
players update the reputation of a third party whenever that individual
participates in an interaction, irrespective of the observer's direct
relationship with them. This process, which relies on interactions involving
others, is referred to as pure indirect reciprocity. Players combine the three
modes of reciprocity, direct, indirect and generalized when \(0 < \lambda < 1\)
and \(0 < \gamma < 1\).

Note that, for \(\gamma = 0\), the model reduces to the model studied by Schmid et
al.~\citep{Schmid:NHB:2021}, where no generalized reciprocity occurs. Our
framework is a natural extension of their work, incorporating generalized
reciprocity to provide a more unified framework of reciprocity. Similar to the
previous model, it can encompass several well-known strategies from both the
direct and indirect reciprocity literature, such as Tit-for-Tat (TFT), Generous
Tit-for-Tat (GTFT), and the simple scoring strategy from indirect
reciprocity~\citep{Berger:GEB:2011}, which assigns good reputations to
cooperators and bad reputations to defectors, while always cooperating in
anonymous interactions. These strategies are represented as \((1,1,0,0,0)\),
\((1,1,q,0,0)\), and \((1,1,0,1,0)\), respectively, using the current notation.
Additionally, our framework can also represent strategies from the generalized
reciprocity literature, such as Alternating Tit-for-Tat
(A-TFT)~\citep{Pfeiffer:PRSB:2005} or Upstream-TFT~\citep{Boyd:SN:1989}, denoted
as \((1,1,0,0,1)\). A-TFT is a strategy that repeats the last action
of their last co-player in the next interaction.

Our model has two salient properties: (i) conditional cooperation by an
individual depends solely on the current standing of their co-player, and (ii)
any reputation update is based on a single observed action within the
population. These properties make the model mathematically tractable, allowing
for the explicit computation of expected payoffs for strategies in any arbitrary
population. To this end, we derive a recursion formula that relates the
probability of a focal player considering a given co-player as ``good'' after
the next interaction to the probability that they currently see them as
``good''. For a complete description of how payoffs are computed, see Sections 1
and 2 in the \SupplementaryInformation.

\subsection*{Nash equilibria and the conditions for cooperation across reciprocities}

\begin{figure}[t!]
    \centering
    \includegraphics[width = \textwidth]{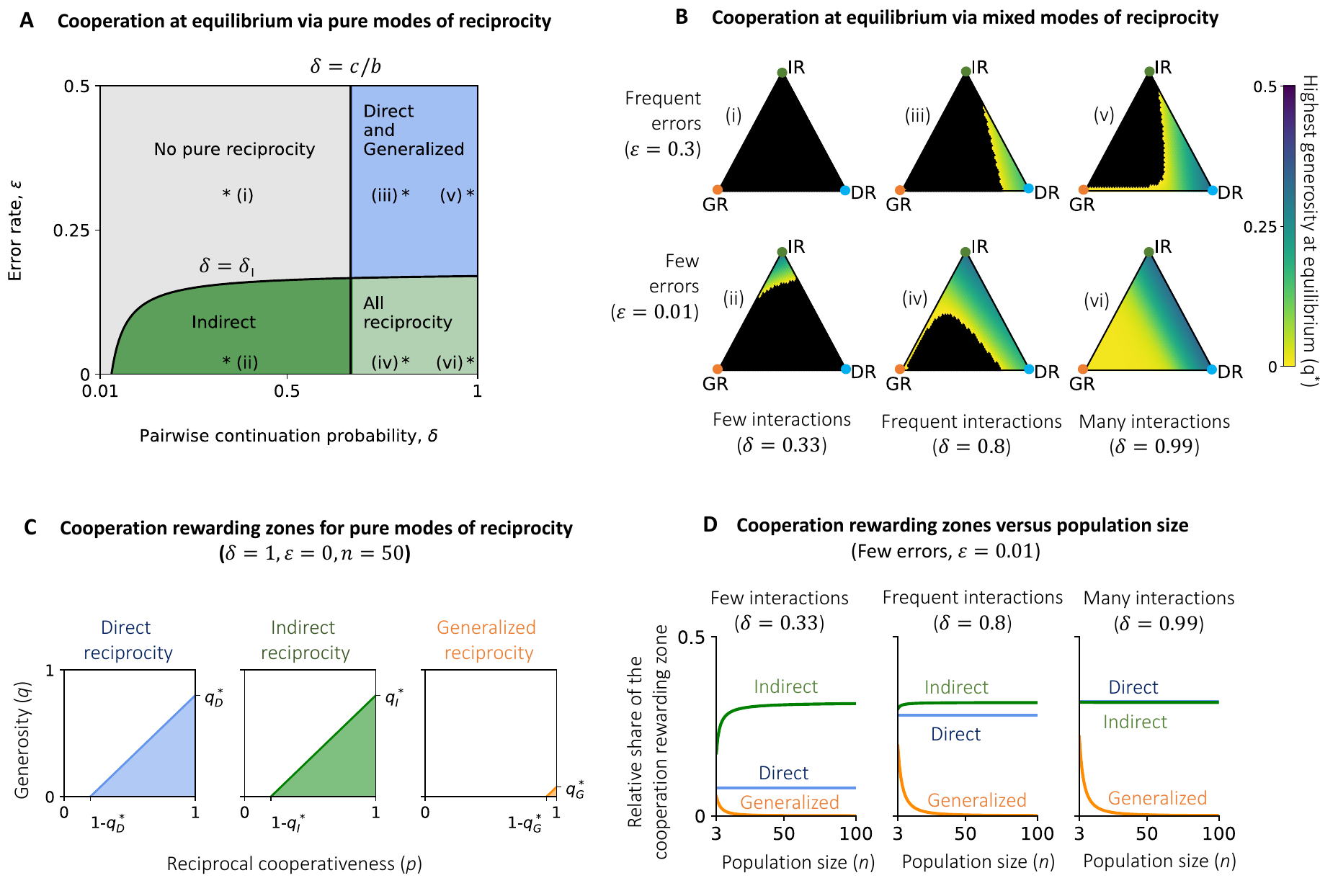}
    \caption
    {\textbf{Equilibrium analysis and cooperation rewarding zones.} (A),
    Pure direct and
    generalized reciprocity can sustain full cooperation at equilibrium if and
    only if the expected number of pairwise interactions is sufficiently high
    (i.e, the pairwise continuation probability, $\delta > c/b$). On the other
    hand, the condition for indirect reciprocity depends both on the expected
    number of pairwise interactions and the rate at which observation errors are
    made: $\delta > \delta_{\mathrm{I}}$; see Eq.~(\ref{Eq:IR-condition-1}).
    \textbf{(B),} All possible mixtures of
    direct, indirect and generalized reciprocity can be represented using a
    simplex wherein each corner represents pure direct (DR), indirect (IR) and
    generalized (GR) reciprocity and points inside (or on the edges) represent
    mixtures of all three (or two) modes of reciprocity. On these simplices, we
    plot the highest generosity towards defection, $q^*$, that sustains full
    cooperation at equilibrium (if full cooperation is feasible). The black
    regions in the simplices indicate that there is no cooperative Nash
    equilibrium for the corresponding mixtures of the parameters (see \Methods{} for more details).
    \textbf{(C),}
    We consider an infinitely long game ($\delta = 1$) where individuals make no
    errors in observation ($\varepsilon = 0$). The cooperation rewarding zones,
    for each mode of reciprocity is annotated with colors in the strategy space
    $(p,q)$. \textbf{(D),} We plot the relative size of the cooperation
    rewarding zone as a fraction of the entire strategy space, for each mode of
    reciprocity as we vary the population size, $n$. Unless otherwise stated, we
    consider a population of 50 individuals. For \textbf{A, B} we
    take benefit to cost ratio, $b/c = 1.5$. For \textbf{C, D} we use $b/c = 5$.}
    \label{fig:equilibrium}
\end{figure}

A strategy is a Nash equilibrium if no player has an incentive to deviate from
it,  given that the other players also follow it.  Our setup allows us to determine whether an arbitrary strategy \(\sigma\) is
a Nash equilibrium. The conditions for this are provided in Section 3 of the
\SupplementaryInformation. Here, we focus on a specific subset of Nash
equilibria that give rise to full cooperation, namely, cooperative Nash
equilibria. For a strategy to be a cooperative Nash equilibrium, it must be both
a Nash equilibrium and self-cooperative. Self-cooperative strategies ensure
cooperation in a homogeneous population, given that there are no errors. In our
framework, a strategy \(\sigma\) is self-cooperative if and only if \(y = p =
1\). We begin by analyzing cooperative equilibria when the population adopts one
of the three pure modes of reciprocity.

The existence of cooperative equilibria depends on the model's parameters. For
both pure direct and generalized reciprocity, full cooperation is sustainable at
equilibrium if, and only if, the expected number of pairwise interactions is
sufficiently high. Specifically, the pairwise continuation probability,
\(\delta\), must satisfy:

\begin{equation}
    \delta > \frac{c}{b}.
    \label{Eq:DR-condition-1}
\end{equation}
\vspace{0.5em}

In the case of pure direct reciprocity (\(\lambda = \gamma = 0\)) with \(y = p =
1\), full cooperation is a Nash equilibrium if

\begin{equation}
    q \leq q^*_{\mathrm{D}} := 1 - \frac{c}{b\delta}.
    \label{Eq:DR-condition-2}
\end{equation}
\vspace{0.5em}

Similarly, for generalized reciprocity (\(\lambda = 0, \gamma = 1\)), full
cooperation is maintained if

\begin{equation}
    q \leq q^*_{\mathrm{G}} := 1 - \frac{c + c\delta(n-2)}{b\delta + c\delta(n-2)}.
    \label{Eq:GR-condition-2}
\end{equation}
\vspace{0.5em}

From conditions~\eqref{Eq:DR-condition-2} and~\eqref{Eq:GR-condition-2}, it can be derived that a fully cooperating population using direct reciprocity can tolerate higher levels of generosity towards errors compared to a population using generalized reciprocity (i.e., \(q^*_{\mathrm{D}} > q^*_{\mathrm{G}}\)). Moreover, this difference becomes increasingly significant in larger populations (i.e., \(q^*_{\mathrm{D}} \gg q^*_{\mathrm{G}}\) when \(n \gg 2\)).

Pure indirect reciprocity allows full cooperation at equilibrium even when the
expected number of pairwise interactions is low. However, this requires
sufficiently reliable observations of third-party interactions. Specifically,
the continuation probability, \(\delta\), and the probability of observational
errors, \(\varepsilon\), must satisfy:

\begin{equation}
    \delta > \delta_\mathrm{I} := \frac{c}{b + (n-2)((1-2\varepsilon)b - c)},
    \label{Eq:IR-condition-1}
\end{equation}
\vspace{0.5em}

Additionally, the strategy must satisfy:

\begin{equation}
    q \leq q^*_{\mathrm{I}}:= 1 - \frac{1 + (n-2)\delta}{1 + (n-2)(1-2\varepsilon)} \frac{c}{b\delta}.
    \label{Eq:IR-condition-2}
\end{equation}
\vspace{0.5em}

We also compare the maximum level of generosity that a fully cooperating
population can adopt at equilibrium when using indirect reciprocity versus the
other modes of reciprocity. When observational errors are sufficiently low, or
the benefit-to-cost ratio is sufficiently high (i.e., when \(b/c >
1/(1-2\varepsilon)\)), a population cooperating via indirect reciprocity can be
more generous at equilibrium than one adopting generalized reciprocity (and
\emph{vice versa}). Moreover, when observational errors are sufficiently low
or the pairwise continuation probability is sufficiently high, a population
cooperating via indirect reciprocity can be more generous at equilibrium than
one using direct reciprocity. In particular, when \(\delta < 1-2\varepsilon\).

The differences in the maximum level of generosity that can be sustained at
equilibrium across the modes of reciprocity offer insight into how each reciprocity promotes
cooperation. To this end, we use the concept of a \textit{cooperation
rewarding zone} from the literature on the repeated prisoner's
dilemma~\citep{Sigmund:book:2010}. For a fixed mode of reciprocity, the
cooperation rewarding zone refers to the set of all strategies that, when
adopted uniformly by a population, make always-cooperate the best response for
any individual (see \Methods{} for details). In
Fig.~\ref{fig:equilibrium}C, we annotate these sets for all three modes of
reciprocity for the infinitely repeated game with no errors (\(\delta = 1,
\varepsilon = 0\)). We observe that the sizes of the cooperation rewarding zones
depend on the mode of reciprocity. In particular, for generalized reciprocity,
the zone is remarkably small. Furthermore, for generalized reciprocity, the size
of the cooperation rewarding zone decreases sharply with an increase in
population size (Fig.~\ref{fig:equilibrium}D), indicating that cooperation is
less likely to be sustained when large populations adopt generalized
reciprocity. In contrast, for direct and indirect reciprocity, the sizes of the
cooperation rewarding zones remain relatively constant across different
population sizes.

Beyond pure reciprocity modes, our equilibrium analysis predicts cooperative
equilibria when individuals use all three forms of reciprocity simultaneously
(see Proposition 1 in the \SupplementaryInformation). We perform numerical
computations to demonstrate this result in Fig.~\ref{fig:equilibrium}B. We observe
that more mixed cooperative Nash equilibria exist when pairwise interactions are
longer and observation errors are unlikely.  Furthermore, regardless of
model parameters,  generosity decreases as more generalized reciprocity is
used.

\subsection*{The evolution of cooperation when the mode of reciprocity is fixed}

The equilibrium analysis has shown that strategies capable of sustaining a
cooperative Nash equilibrium exist for all three modes of reciprocity.
However, it remains unclear whether
these strategies would evolve if individuals can adapt their strategies over
time through social learning.
To explore this, we conduct evolutionary simulations where individuals update
their strategies by imitating others in the population. Additionally, occasional
mutations introduce new strategies into the population.
We first consider simulations where everyone in the population is restricted to
the same mode of reciprocity (i.e., pure direct, indirect, or generalized
reciprocity). As a result, players are only able to evolve the strategy
components \(y\), \(p\), and \(q\). Furthermore, we assume that mutations are
relatively rare, leading to mostly homogeneous populations where all players
adopt the same strategy. When a mutant appears, it either becomes fixed in the
population (becoming the new resident) or goes extinct. This scenario is known
as the limit of rare exploration (or the limit of rare
mutation)~\citep{Imhof:PRSB:2010, baek2016comparing}. For a detailed description
of the evolutionary process, see \Methods.

\begin{figure}
    \centering
    \includegraphics[width = \textwidth]{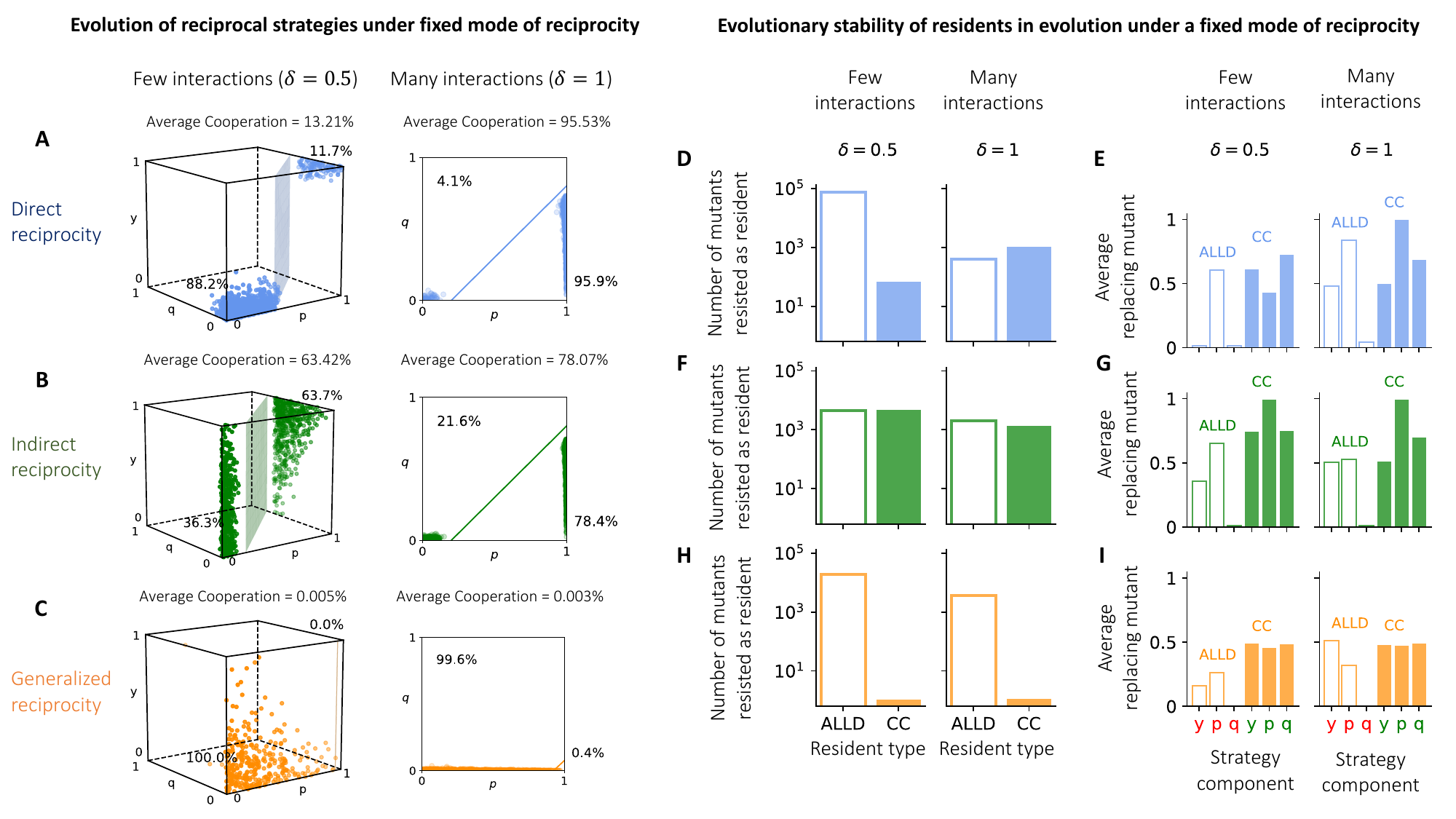}
    \caption{\textbf{Evolutionary dynamics at rare exploration under a single mode of pure reciprocity.} We perform evolutionary simulations in the limit of rare exploration wherein individuals were restricted to use either pure direct reciprocity (A,  D,  E), pure indirect reciprocity (B,  F,  G) or pure generalized reciprocity (C,  H,  I). In the evolutionary process, individuals were free to choose the strategy elements $y$, $p$ and $q$ from the unit cube. \textbf{(A),  (B),  (C):} In these plots, the dots indicate the most successful residents in terms of resisting invasion from mutants. In the case of direct and indirect reciprocity, the most successful residents form two major clusters in the strategy space - they are either in the vicinity of the strategy that always defects, $y \approx p \approx 0$, or in the vicinity of the strategy that cooperates conditionally: $p \approx 1$. While using generalized reciprocity, successful strategies almost always defect $y \approx q \approx 0$.
    The percentages represent the time spent by the whole evolutionary process on each side of the solid plane (or line) which contains all equalizer strategies. At $\delta = 1$, the effect of cooperation in the first round $y$, is negligible and therefore the strategy space is sufficiently represented by a square instead of a cube. \textbf{(D),  (F),  (H):}  For theses plots, we conduct simulations to study the evolutionary robustness of cooperation and defection under each mode of reciprocity (for the exact details of the strategies \textrm{ALLD} type and \textrm{CC} type see \Methods). Effectively, we measure time with the number of mutants the residents resist before being invaded by a mutant. \textbf{(E),  (G),  (I):} In here, we plot the strategy components of average mutant that successfully invade the \textrm{ALLD} type and \textrm{CC} type residents in evolutionary simulations. Parameters: $n = 50, b = 5, c = 1, \beta = 10, \varepsilon = 0$ and $\delta \in \{0.5,1\}$.}
    \label{fig:single-reciprocity-evolution}
\end{figure}

In simulations where players have very few interactions (\(\delta = 0.5\)),
cooperative Nash strategies are adopted by the population only in the case of
indirect reciprocity. For the other two modes of reciprocity, the population
predominantly adopts defecting strategies. However, as the number of
interactions increases (\(\delta = 0.99\)), direct reciprocity leads to a
predominantly cooperative population, surpassing the cooperation rate of
indirect reciprocity. In contrast, generalized reciprocity continues to result
in defecting strategies. Thus, consistent with earlier models~\citep{nowak2007upstream,
Pfeiffer:PRSB:2005}, we find that generalized reciprocity alone is insufficient
to produce cooperation in large populations. Further simulations (summarized in
Fig.  S2) confirm that generalized reciprocity is unable to stabilize
cooperation across varying population sizes, selection strengths, and
benefit-to-cost ratios.

To further clarify these results, we compare the evolutionary robustness of
cooperative and defecting strategies for each type of reciprocity. We focus on
two strategies: one that almost always reciprocates cooperation (CC) and one
that almost always defects (ALLD). Specifically, the ALLD strategy is defined as
\((0.01, 0.01, 0.01, \lambda, \gamma)\), while the CC strategy is defined as
\((0.99, 0.99, 0.5, \lambda, \gamma)\). 
We consider a homogeneous resident population that adopts either the CC or ALLD
strategy. We then simulate the evolutionary process by repeatedly introducing
random mutant strategies and record how many mutants, on average, are needed
until one successfully replaces the resident. Additionally, we record the
properties of the mutant strategies that successfully invade these populations.
The results are summarized in Fig.~\ref{fig:single-reciprocity-evolution}D-I.

For indirect reciprocity, the robustness of cooperative and defective strategies
is similar, regardless of the number of interactions. \textrm{ALLD} is always
replaced by more cooperative mutants, while \textrm{CC} is replaced by Generous
TFT like strategies. In contrast, for direct reciprocity, the \textrm{CC} strategy
is more stable than \textrm{ALLD} when the number of interactions is
high, but the opposite is true when the number of interactions is low. This
explains why the population adopts defecting strategies for small \(\delta\).
When players rely solely on generalized reciprocity, the \textrm{CC} strategy is
highly unstable, regardless of game length. In both short and long games,
\textrm{CC} is quickly replaced by almost any mutant due to its excessive
generosity toward defectors (Fig.~\ref{fig:single-reciprocity-evolution}I). In
contrast, \textrm{ALLD} remains stable under generalized reciprocity.

\subsection*{The co-evolution of direct, indirect and generalized reciprocity}

The analysis in the previous section highlights when cooperation is likely to
evolve when individuals rely on a single mode of reciprocity. However,
reciprocal strategies themselves are subject to selection. We now explore a more
general scenario where individuals can evolve their strategy elements \(y\),
\(p\), and \(q\), while also choosing between one of the three pure modes of
reciprocity. 
We perform evolutionary simulations in the limit of low mutations
and compare the outcomes of our unified model with those of Schmid et al. to
understand the impact of introducing generalized reciprocity. We refer to
Schmid's model as the ``reduced model'' and ours as the ``extended model''. 
In the reduced model, the population primarily relies on indirect reciprocity,
except when observation errors are high, making third-party observations less
reliable.  In the extended model,  however,  the population predominantly adopts generalized 
reciprocity, which leads to lower cooperation rates.  Across most of the parameter space,  by varying \(\delta\) (the
expected number of interactions) and \(\varepsilon\) (the probability of
errors), we observe a higher cooperation rate in the reduced model in comparison to the extended model.

\begin{figure}[t!]
    \centering
    \includegraphics[width = \textwidth]{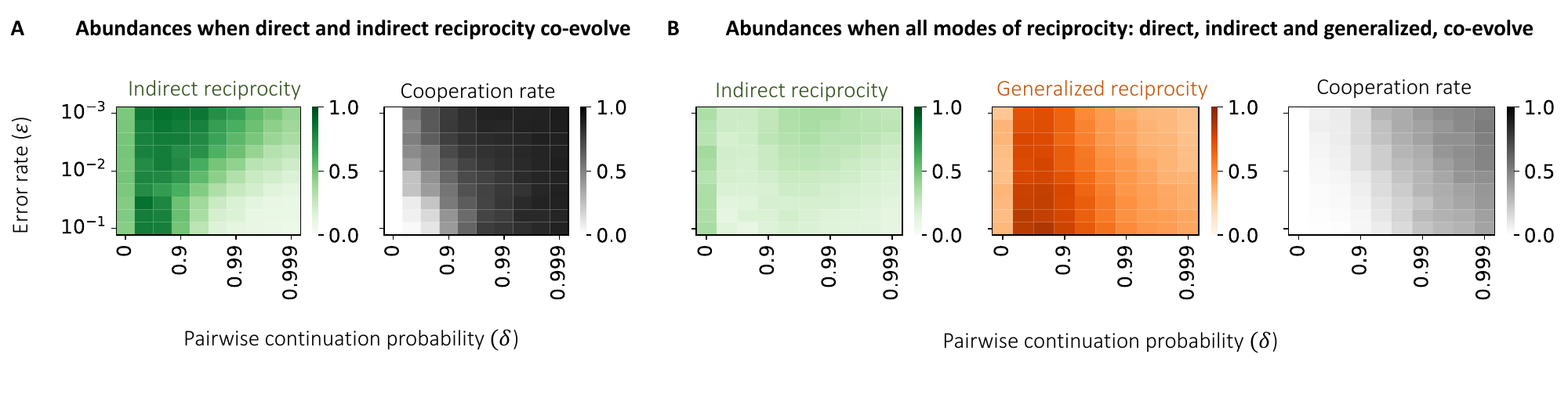}
    \caption{\textbf{A comparison between evolutionary simulations with and without generalized reciprocity demonstrates its effect on the evolution of cooperation.} In the evolutionary simulations that we conduct for this figure, individuals can choose to adopt different modes of reciprocity. In \textbf{(A)}, individuals are restricted to adopt either direct reciprocity or indirect reciprocity while in \textbf{(B)}, individuals can choose between the three modes of reciprocity. The evolutionary process is conducted for rare exploration. See \Methods{} for more details on the evolutionary simulations. We present the time average of the abundance of each mode of reciprocity and report how abundant cooperation is in the evolving population. We present these averages as we vary the parameters of the game, the pairwise continuation probability $\delta$ and the observation error rate, $\varepsilon$. All other parameters of the game and the evolutionary process are the same as in Fig.~\ref{fig:single-reciprocity-evolution}} 
    \label{fig:evolution-with-and-without-GR}
\end{figure}

To understand why the inclusion of generalized reciprocity undermines the
evolution of cooperation in the limit where individuals rarely explore, we
closely examine one particular simulation (i.e., where \(\delta = 0.999\) and
\(\varepsilon = 10^{-3}\)). For these parameters, we observe that while in the
reduced setup cooperation is abundant via direct reciprocity, in the extended
setup, individuals adopt direct (38\%), indirect (32\%), and generalized (30\%)
reciprocity equally often (see Fig.  S4A). In
particular, they cooperate rarely when adopting generalized reciprocity. By
studying the evolutionary transitions in this example, we identify that the key
mechanism behind the breakdown of cooperation in the extended setup is neutral
selection (see Fig.  S4B). In the extended setup, populations that
cooperate via direct reciprocity are neutrally invaded by mutants that also
cooperate strongly, but via generalized reciprocity. After invasion, these
cooperative pay-it-forward residents are themselves susceptible to invasion by
defectors. Once defection becomes prevalent, generalized reciprocity is strongly
favored. Afterward, transitions to any mode of reciprocal cooperation are
unlikely.

\begin{figure}[t!]
    \centering
    \includegraphics[width = \textwidth]{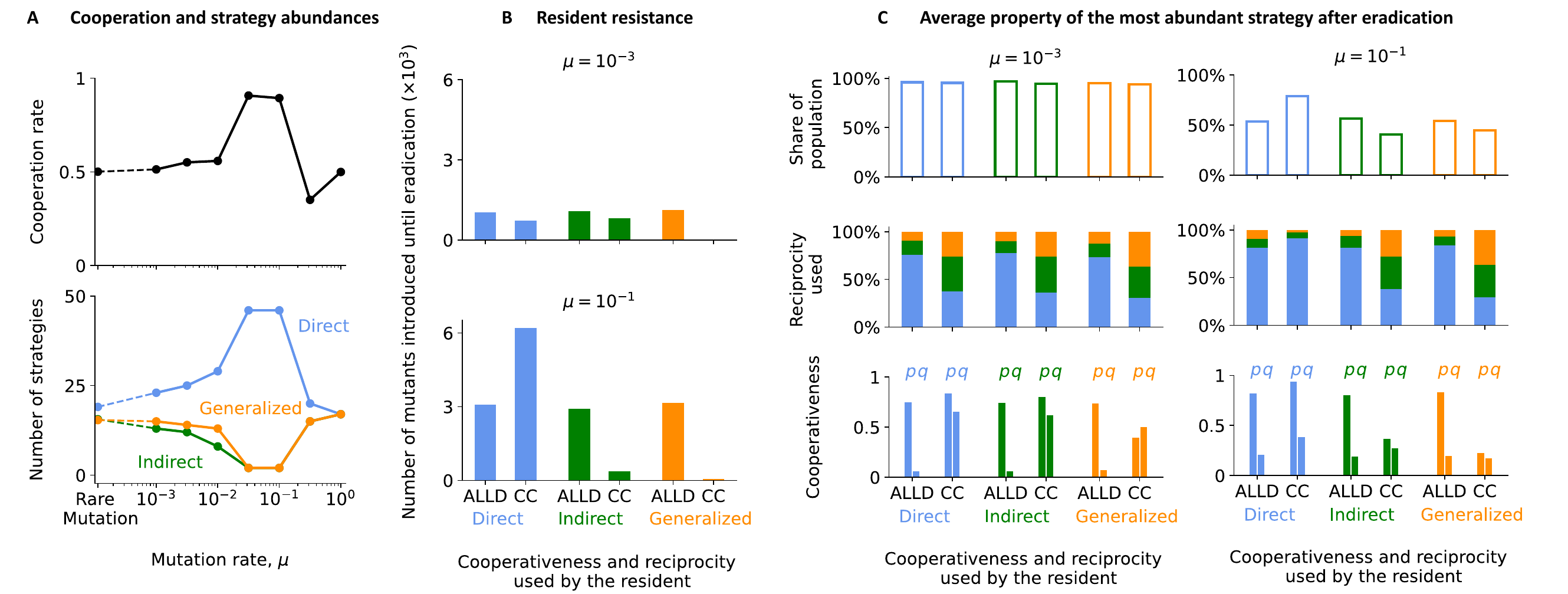}
    \caption{\textbf{The co-evolution of reciprocal strategies, when exploration
    is frequent, leads to the emergence of stable and abundant cooperation.} We
    conduct evolutionary simulations in which individuals frequently adopt
    random strategies thereby leading to mixed populations.
    \textbf{(A),} For each simulation we perform, we fix
    the rate at which individuals adopt a random strategy, $\mu$ (also called
    the mutation or the exploration rate). From each simulation, we report the
    average cooperation rate and the average composition of the population, with
    respect to the reciprocal strategies that individuals adopt. At intermediate
    values of $\mu$ we observe that cooperative individuals who adopt direct
    reciprocity dominate the evolving population. \textbf{(B),} We conduct an
    invasion experiment to study how resistant reciprocal strategies are to the
    process of mutation and selection. To this end, we take six different
    resident populations, two for each mode of reciprocity -- \textrm{ALLD}
    and \textrm{CC} (see \Methods{} for exact details) and report
    the average number of mutations necessary to completely eradicate the
    incumbent population. We conduct these simulations for two values of $\mu$:
    $10^{-3}$ and $10^{-1}$. When exploration is intermediate ($\mu = 10^{-1}$),
    the cooperative direct reciprocity population displays strong resistance.
    \textbf{(C)}, From the simulations we perform to study the invasion
    resistance, we report the average property of the most abundant strategy
    after the eradication is complete. We report the share of the population
    that the most abundant strategy occupies on average (top row), the estimated
    likelihood of the reciprocal strategy that it uses (middle row) and the
    average values of $p$ and $q$ of its strategy (bottom row). Parameters:
    $\delta = 0.999$, $\varepsilon = 0$, $b,c$ and $\beta$ are kept the same from
    Fig.~\ref{fig:single-reciprocity-evolution}}.
    \label{fig:frequent-mutation}
\end{figure}

\subsection*{Co-evolution of reciprocities when individuals frequently explore new strategies}

In the limit where individuals rarely explore new strategies, our analysis
predicts weak reciprocal cooperation. To study the role that exploration plays
in the evolution of reciprocal cooperation, we perform evolutionary simulations
in which individuals randomly adopt arbitrary strategies at a fixed rate, $\mu$
(Fig. \ref{fig:frequent-mutation}). In these simulations, we study outcomes
for long games ($\delta = 0.999$). Unlike simulations for rare exploration, here
populations are often composed of more than two strategies. For a detailed
description of this evolutionary process see \Methods.

Our analysis
identifies that an intermediate rate exploration has a positive effect on the
evolution of cooperation, even in the presence of generalized reciprocity. When
the rate of exploration is neither too high nor too low, individuals strictly
prefer to cooperate via direct reciprocity. In contrast, when
rates of exploration are either too high or too low, individuals try all modes
of reciprocity about equally but cooperation rates drop to approximately 50\%
(Fig.~\ref{fig:frequent-mutation}A).

We perform additional simulations to
examine the reason why cooperation via direct reciprocity dominates at
intermediate rates of exploration (Fig.~\ref{fig:frequent-mutation}B,C). We
conduct evolutionary simulations starting with a unique homogenous resident
population each time and study (i) the average time it takes to eradicate the
resident strategy completely and (ii) the average property of the most abundant
strategy after the eradication is complete. We observe that when individuals
explore at an intermediate rate ($\mu = 10^{-1}$), populations which use direct
reciprocity and cooperate highly take the longest to eradicate (on average).
Additionally, when they are indeed eradicated, the resulting population is
usually abundant with individuals who also use direct reciprocity to cooperate.
In contrast, other resident populations are relatively less stable to this
process of exploration and learning. Moreover, if the initial resident
population cooperates lowly, they are usually replaced by a population that is
abundant with individuals who cooperate via direct reciprocity. At intermediate
rates of exploration, cooperation through direct reciprocity exhibits greater
stability and is more prone to positive selection when contrasted with other
kinds of reciprocal behaviour.

\section*{DISCUSSION}

In repeated interactions, humans often use past outcomes to guide current
decisions~\citep{Molleman:PRSB:2016,Milinski:PRSB:2001,rockenbach2006efficient,Okada:SciRep:2018,Sommerfeld:PNAS:2007}.
They may recall how others have treated them (generalized reciprocity), their
own past interactions with the same person (direct reciprocity), or their
knowledge of how the person treated others (indirect reciprocity). Despite the
different ways humans reciprocate, it remains unclear when one form of
reciprocity is preferred over another. A recent study~\citep{Schmid:NHB:2021}
explored this question using a model limited to direct and indirect reciprocity.
Here, we extend this work by including generalized reciprocity. We derive a
method to compute payoffs when players combine all three strategies and
characterize the resulting Nash equilibria, providing explicit conditions for
each form of reciprocity to sustain cooperation.

We find that generalized reciprocity alone can stabilize cooperation in
populations expecting frequent interactions. However, generalized reciprocity
often undermines cooperation when players imitate others, leading to widespread
defection (Fig.~\ref{fig:evolution-with-and-without-GR}). Evolutionary
simulations show that populations reliant on generalized reciprocity are
vulnerable to defection, ultimately collapsing into non-cooperative behavior
(Fig.  S4). Stable cooperation is only maintained when players explore
new strategies at a steady rate. This observation resonates with a recent
finding which also reports a strong correspondence between exploration and
cooperation via direct reciprocity \citep{tkadlec2023mutation}. 

Our results on the limits of generalized reciprocity
(Fig.~\ref{fig:single-reciprocity-evolution}C) agree with previous
models~\citep{nowak2007upstream,rankin2009assortment,hamilton2005contingent,Pfeiffer:PRSB:2005}.
Cooperation through generalized reciprocity alone rarely emerges in random
interactions without additional mechanisms like
assortment~\citep{rankin2009assortment}, group
movement~\citep{hamilton2005contingent}, or immediate reciprocal
payback~\citep{nowak2007upstream}. 

However, our view on the emergence of generalized reciprocity diverges from
earlier work. While prior models explain why cooperation can evolve through
generalized reciprocity with supporting mechanisms, they do not explain why
generalized reciprocity evolves in the first place. Our results show that when
competing forms of reciprocity are allowed, generalized reciprocity is rarely
favored
(Figs.~\ref{fig:evolution-with-and-without-GR},~\ref{fig:frequent-mutation}).
Even if additional mechanisms favor it, it is unclear why direct or indirect
reciprocity wouldn't be equally selected, as seen in network-based models of
direct reciprocity~\citep{ohtsuki2007direct,Pacheco:JTB:2008}.

Nowak and Roch~\citep{nowak2007upstream} suggest that direct reciprocity can
promote generalized reciprocity, but our results show the opposite. Players tend
to avoid generalized reciprocity, favoring defection when direct and generalized
reciprocity compete (Fig.  S5). We attribute this discrepancy to
differences in how direct reciprocity is modeled. In our study, individuals
choose between their general experience and specific interactions, whereas in
Nowak and Roch's model, direct reciprocity only occurs after receiving
cooperation, functioning more like a reward system~\citep{sigmund2001reward}.

Sasaki et al.~\citep{sasaki2023evolution} propose a 
strategy that simultaneously applies generalized and indirect reciprocity while
interacting in an infinitely large population. When a player uses this strategy,
they invariably cooperate after receiving a cooperation. In case they do not
receive a cooperation, they cooperate based on the reputation of their
beneficiary (which could be assigned with a simple scoring rule). The authors
refer to this conditional strategy as integrated reciprocity. They observe that
integrated reciprocity is able to stably coexist with unconditional defectors in
games that are infinitely long. Even though defectors take advantage of
individuals who use integrated reciprocity, the strategy does not die out
because rewards for being good' (from other integrated reciprocators) balance
out the exploitation. This result indicates that a coexistence between
cooperative generalized reciprocity and cooperative indirect reciprocity may
stably emerge in our model too. However, we do not detect any such stable
coexistence in our evolutionary simulations (Fig.~\ref{fig:frequent-mutation}A).
We argue that this is because players can sometimes adopt random strategies in
our evolutionary process, and more importantly, these strategies can utilise
direct reciprocity (a strategy not included in Sasaki et al.). When games are
long, cooperation via direct reciprocity is very likely to spread in any
population that uses a combination of generalized and indirect reciprocity. 

Generalized reciprocity is often argued to evolve due to its lower cognitive
demands~\citep{barta2011cooperation,Pfeiffer:PRSB:2005,Rutte:PLoSB:2007},
compared to direct or indirect
reciprocity~\citep{Milinski:PNAS:1998,duffy2014cognitive,Santos:Nature:2018}. We
do not account for cognitive costs in our model. However, we note that it is
plausible that once (in)direct reciprocity is considered costly in this model,
the emergence of defection via generalized reciprocity becomes even more likely.
Therefore, it is still not clear how exactly individuals with limited memory
would learn to cooperate.

Lastly, we address the apparent conflict between our findings and empirical
studies that show humans engage in generalized
reciprocity~\citep{baker2014paying,melamed2020robustness,herne2013experimental,simpson2018roots,stanca2009measuring}.
We propose that pay-it-forward cooperation may have evolved independently of
direct or indirect reciprocity. Our model serves as a counterfactual where all
forms of reciprocity co-evolve to maximize payoffs. In this scenario, we would
expect minimal reciprocal cooperation or generalized reciprocity to be
undermined by direct reciprocity, as observed. This suggests that generalized
reciprocity may have emerged via mechanisms beyond material payoffs. Supporting
this view, a neurological study~\citep{watanabe2014two} shows that
pay-it-forward cooperation is driven by empathy, while reputation-based
reciprocity engages self-centered cognition, hinting that generalized
reciprocity may evolve through non-material motivations.

\section*{MATERIALS AND METHODS}

\subsection*{Cooperation rewarding zone}

We consider that a population of size $n$ individuals where everyone uses the
strategy $\sigma := (y,p,q,\lambda,\gamma)$. Then, the resident strategy
$\sigma$ is said to belong in the cooperation rewarding zone if a single
individual can best-respond by unilaterally switching to a strategy that
cooperates in every round. If a single individual's best-response is to switch
to always-defect, then $\sigma$ is said to belong in the defection rewarding
zone. If $\sigma$ belongs in neither the cooperation nor the defection rewarding
zone, it belongs in the set of equalizer strategies. In this case, any strategy
is a best-response for a single unilaterally deviating individual. These three
sets are disjoint and their union is the entire strategy space, $[0,1]^5$. If
the resident population uses pure direct, pure indirect or pure generalized
reciprocity, the cooperation rewarding zone is the subset of the strategy space
in which the strategy components $p$ and $q$ satisfy $p-q > 1 - q^*$. Here $q^*$
is either $q^*_\mathrm{D}$, $q^*_\mathrm{I}$ or $q^*_\mathrm{G}$, from Eqs.
(\ref{Eq:DR-condition-2}), (\ref{Eq:IR-condition-2}) and
(\ref{Eq:GR-condition-2}) respectively, depending on the mode of reciprocity
employed by the resident population. If $p-q < 1-q^*$, the strategy belongs in
the defection-rewarding zone. If $p-q = 1-q^*$, the strategy is an equalizer
strategy.

\subsection*{Evolutionary dynamics at any arbitrary mutation rate}

At each time step of our evolutionary process, a focal individual is picked
uniformly at random from the population to update their strategy. With
probability $\mu$, the focal player picks a strategy uniformly at random from
the set of permissible strategies, and adopts it. With probability
$1-\mu$, the focal player picks another individual from the population,
uniformly at randon (the role model), and adopts their current strategy with
probability

\begin{equation}
    \rho_{\mathrm{F} \to \mathrm{R}} = \displaystyle \frac{1}{1 + e^{\beta(\pi_\mathrm{F} - \pi_\mathrm{R})}}.
    \label{Eq:focal-role-model-imitation}
\end{equation}

Here $\pi_\mathrm{F}$ and $\pi_\mathrm{R}$ are the expected payoffs of the focal
and the role model individuals in the current population. The parameter $\beta$
is called the selection strength. It reflects the importance that players assign
to payoff differences while imitating a new strategy. The computation of
expected payoffs of individual strategies at any arbitrary population
composition is discussed in \SupplementaryInformation.

\subsection*{Evolutionary dynamics at rare exploration}
The evolutionary process in the limit of rare exploration ($\mu \to 0$) assumes
that a single mutant in a homogenous population either undergoes fixation or
goes extinct before the next mutant arises. In other words, the expected
time-interval between the arrival of two mutants is sufficiently larger than the
expected time of fixation (or extinction) of a single mutant. This assumption
largely simplifies the evolutionary process. With this, one can assume that the
population is almost always composed by a single strategy. The entire history of the
evolutionary process can be summarized by recording a) the sequence of resident
strategies that occupied the population and b) the length of evolutionary time
that each resident strategy occupied the population (in other words, the number
of mutants they resisted an invasion from). The outcome of each simulation is
thus a sequence $\{ \sigma_1, \sigma_2, \sigma_3, ...., \sigma_l, ....\}$ of
resident strategies, together with a list $\{t_1, t_2, t_3, ...., t_l, ....\}$
that denotes how long each resident was present.

For the mutant fixation dynamics, we use the imitation process described in the previous section. In this dynamical process, at every time step, a pair of focal and role model players are picked uniformly at random from the population. The focal player compares their current expected payoff $\pi_\mathrm{F}$, with the current expected payoff of the role model, $\pi_\mathrm{R}$, and switches to the role model's strategy with probability, $\rho_{\mathrm{F} \to \mathrm{R}}$ given in Eq. (\ref{Eq:focal-role-model-imitation}). The probability that a single mutant $\sigma_M$, attains fixation when introduced in a homogenous $\sigma_R$-playing population is computed exactly using the formula derived in Traulsen et al.~\citep{traulsen2006stochastic},
\begin{equation}
    f(\sigma_M, \sigma_R) =  \frac{1}{1 +\displaystyle \sum_{m=1}^{n-1} \prod_{k=1}^m e^{-\beta(\pi_M(k) - \pi_R(k))}}.
    \label{Eq:fixation-prob-imitation}
\end{equation}

Here, $\pi_M(k)$ and $\pi_R(k)$ are respectively the expected payoffs of the mutant and the resident strategy when $k$ individuals play $\sigma_M$ and $n-k$ individuals play $\sigma_R$. We compute this probability in our simulations to determine whether a newly introduced mutant takes over the population.

\subsection*{Methods for figures}
\textbf{Methods for Fig.~\ref{fig:equilibrium}}, In
Fig~\ref{fig:equilibrium}A, we consider a population of 50 individuals and plot
the regions in the $\delta$-$\varepsilon$ plane corresponding to regions where
strategies employing pure direct, pure indirect and pure generalized reciprocity
are able to constitute a cooperative Nash equilibria. The regions are marked by
the two boundaries: $\delta = \delta_{\mathrm{D}} = \delta_{\mathrm{G}}$  and
$\delta = \delta_{\mathrm{I}}$. The expressions for $\delta_{\mathrm{D}}$,
$\delta_{\mathrm{I}}$ and $\delta_{\mathrm{G}}$ are given in
Eq.~(\ref{Eq:DR-condition-1}) and Eq.~(\ref{Eq:IR-condition-1}).

In Fig~\ref{fig:equilibrium}B, we use the simplex representation of reciprocal
strategies to denote whether cooperative strategies of the form $(1,1,q,\lambda,
\gamma)$ are able to constitute a Nash equilibria. Each point inside the
triangular simplex $\mathcal{S}_{\mathrm{DIG}}(n)$ represents a tuple of three
complimentary probabilities $(\alpha_\mathrm{D}, \alpha_\mathrm{I},
\alpha_\mathrm{G})$ and corresponds to a unique pair $(\lambda, \gamma) \in
[0,1]^2$. The probabilities, $\alpha_\mathrm{D}, \alpha_\mathrm{I}$ and
$\alpha_\mathrm{G}$ represent the effective likelihood with which individuals in
a homogenous population use direct, indirect and generalized reciprocity
respectively. For a precise mathematical
definition of these probabilities, see the section on `effective likelihood of
reciprocity' in \SupplementaryInformation. The three corners of the triangular
simplex correspond to three modes of pure reciprocity. Pure direct ($\lambda =
0$, $\gamma = 0$), pure indirect ($\lambda = 1$, $\gamma = 0$) and pure
generalized ($\lambda = 0$, $\gamma = 1$) reciprocity correspond to the corners
of the simplex: $(1,0,0)$, $(\alpha_\mathrm{D}^{min},\alpha_\mathrm{I}^{max},0)$
and $(\alpha_\mathrm{D}^{min},0,\alpha_\mathrm{G}^{max})$ respectively. These
corners are marked by \textrm{DR}, \textrm{IR} and \textrm{GR} in Fig.
~\ref{fig:equilibrium}B. The relative distance of a point from the corners of
the simplex represents how often individuals use different kinds of reciprocity
to update reputations in the corresponding homogenous population. For several
discrete points inside the simplex, we numerically check whether there is a $q
\in [0,1)$ such that $(1,1,q,\lambda,\gamma)$ is a Nash equilibrium strategy.
The black regions on the simplices correspond to populations which cannot
constitute a fully-cooperative Nash equilibria. When cooperative Nash equilibria
exist, we plot the highest value of $q$ (generosity towards defection), $q^*$,
that sustains full cooperation in equilibrium.

In Fig~\ref{fig:equilibrium}C we consider a a game with $n=50$ players that is infinitely long ($\delta = 1$) and in which players do not make any observational errors ($\varepsilon = 0$). For each mode of reciprocity, we annotate regions in the strategy space $(p,q)$ in which every strategy rewards cooperation (or equivalently, falls in the cooperation rewarding zone. See the section on cooperation rewarding zone for more details). The cooperation rewarding zone for a fixed mode of reciprocity is a triangle subtended by the points $(1-q^*,0), (1,q^*)$ and $(1,0)$ in the strategy space $(p,q)$ where $q^*$ is $q^*_\mathrm{D}$, $q^*_\mathrm{I}$ or $q^*_\mathrm{G}$ depending on the mode of reciprocity. In Fig~\ref{fig:equilibrium}D we plot the sizes of these triangles, relative to the whole strategy space, for a game in which players make occasional errors ($\varepsilon = 0.01$), as we vary the population size $n$.  


\noindent \textbf{Methods for Fig.~\ref{fig:single-reciprocity-evolution}},
For this figure, we conduct evolutionary simulations at rare exploration
(mutations). For panels A, D and E, we constrain the permissible
set of strategies to $\lambda = \gamma = 0$ (pure direct reciprocity).
Similarly, for panels B,  F and G we constrain it to $\lambda = 1, \gamma = 0$
(pure indirect reciprocity) and for panels C,  H and I to $\lambda = 0, \gamma =
1$ (pure generalized reciprocity). In the simulations for panels A,  B,  C,
$10^7$ mutants attempted to invade resident populations over the course of the
resident replacement process. The simulations began with a random strategy
$(y,p,q,\lambda,\gamma)$ where $(y,p,q)$ was uniform randomly drawn from the set
$[0,1]^3$ and $\lambda, \gamma$ were set to either $(0,0)$, $(1,0)$ or $(0,1)$
depending on the mode of reciprocity being studied. Subsequent mutants were
drawn analogously. Over the course of a single simulation, we kept the
parameters fixed at $\varepsilon = 0, c = 1, b = 5, \beta = 10$ and $\delta \in
\{0.5,1\}$.

We use colored dots to denote the 1000 residents that were most successful in resisting invasions from mutants. For the simulation at $\delta = 0.5$ in panel C,  we show only 329 (and not 1000) residents, because the process only generated 329 different residents. In every plot in panels A,  B and C, the solid plane (or the solid line) in the strategy space $(y,p,q)$ denotes the subspace that contains the equalizer strategies~\citep{Press:PNAS:2012}. As residents, the equalizer strategies, set the payoff of any arbitrary mutant to a fixed value. As a result, these strategies also constitute Nash equilibria (see Theorem~1
in \SupplementaryInformation). These strategies satisfy $p-q = 1- q^*$. The
value of $q^*$ depends on the mode of pure reciprocity: it is $q^*_{D}$ for pure
direct reciprocity, $q^*_I$ for pure indirect reciprocity and $q^*_G$ for pure
generalized reciprocity. See
Eqs.~(\ref{Eq:DR-condition-2}),~(\ref{Eq:IR-condition-2})
and,~(\ref{Eq:GR-condition-2}).

For panels D-I, we conduct multiple independent evolutionary simulations that begin with either an \textrm{ALLD} type or a \textrm{CC} type strategy. Here, the \textrm{ALLD} type strategy is $(0.01, 0.01, 0.01, \lambda, \gamma)$ while the \textrm{CC} type is $(0.99,0.99,0.5,\lambda,\gamma)$. Each simulation stops when the initial population is taken over by a mutant. We record the number of mutants the initial population resisted before they are invaded. The average of these numbers across 1000 independent simulations are plotted in panels d,f and h. We also record the strategy that successfully invades these two residents in every evolutionary simulation. The average strategy that successfully invades, across 1000 independent simulations, are shown in panels e,g and i. 

\noindent \textbf{Methods for Fig.~\ref{fig:evolution-with-and-without-GR}}, In this figure we compare results from two types of evolutionary simulations. In the first type, strategies that individuals can adopt are restricted to pure direct and pure indirect reciprocity (panel a). That is, they adopt from the set
\begin{equation}
    \Sigma_{\mathrm{DI}} := \Big\{(y,p,q,\lambda,0) \; \Big|  \; (y,p,q) \in [0,1]^3 \land \lambda \in \{0,1\} \Big\}.
\end{equation}

In the second type, they are free to choose between all the three types of reciprocity (panels b). That is, they choose from, 
\begin{equation}
    \Sigma_{\mathrm{DIG}} := \Big\{(y,p,q,\lambda,\gamma) \; \Big| \; (y,p,q) \in [0,1]^3 \land (\lambda,\gamma) \in \{(0,0),(1,0),(0,1)\} \Big\}
    \label{Eq:Sigma_DIG}
\end{equation}

All simulations are conducted in the limit of rare exploration. The parameters $\delta$, $\varepsilon$, $b$, $c$ and $\beta$ remain fixed over the course of a single simulation. Each simulation is ran until $10^7$ mutants. We calculate average properties of the resident strategies that occupy the population in these simulations. We plot the time average of the likelihood that individuals use indirect reciprocity and generalized reciprocity to update reputations (for more details on these computations see the section on effective likelihood of reciprocity in \SupplementaryInformation). We also compute the likelihood that individuals cooperate in pairwise interactions. The cooperation rate in each homogenous population is calculated using the result in Proposition~1 in \SupplementaryInformation. In Fig.  S3, we plot the average time series of these likelihoods and cooperation rate for specific values of the parameters $\delta$ and $\varepsilon$.

\noindent \textbf{Methods for Fig.~\ref{fig:frequent-mutation}}, For this figure we conduct evolutionary simulations at arbitrary mutation rates. For all simulations we perform, we fix the set of permissible strategies $\Sigma$ to $\Sigma_{\mathrm{DIG}}$ as given by Eq. (\ref{Eq:Sigma_DIG}). All simulations begin with a homogenous population in which everyone adopts the strategy that defects in every round $(0,0,0,\lambda,\gamma)$. For figures in panel \textbf{A}, we ran our simulations for $10^7$ time steps when $\mu \in \{10^-3,10^{-2.5}\}$ and for $10^6$ times steps when $\mu \in \{10^{-2},10^{-1.5},10^{-1},10^{-0.5},1\}$. For each simulation we conduct, we report the expected frequency of cooperative interactions, averaged over all populations that emerge in the process. We also report the relative frequency of the three modes of reciprocity, averaged over entire simulations. For figures in panel \textbf{B}, we still keep $\Sigma = \Sigma_{\mathrm{DIG}}$ but start with different initial populations. We study six initial populations, two for each mode of reciprocity: the \textrm{ALLD} and the \textrm{CC} type (similar to Fig. \ref{fig:single-reciprocity-evolution}). Here the \textrm{ALLD} type is a strategy of the form $(0.01,0.01,0.01,\lambda,\gamma)$ while the \textrm{CC} type is $(0.99,0.99,0.4,\lambda,\gamma)$. For each initial population, we conduct 1000 independent simulations. We continue these simulations until the first time all individuals adopt a different strategy than the one they assumed at the beginning of the simulation. We estimate the time it takes for this event to occur by averaging over the independent simulations. Across simulations, we also compute average composition of the population when this event occurs. 
\\ \\
\noindent
{\bf Acknowledgements.} 
This work was supported by the European Research Council Starting Grant 850529 (E-DIRECT) and by the Max Planck Society. 
 \\ \\
\noindent
{\bf Author contributions.} 
S. P., N. G, and C. H. designed the research; 
S. P. performed the research; 
S. P. wrote the manuscript;
S. P.,  N. G, and C. H. edited the manuscript; \\

\noindent
{\bf Competing interests.} The authors declare no conflict of interest. \\

\noindent
{\bf Data and materials availability.} All code used to carry out the
evolutionary simulations is available in an online repository, and the data used
to produce the manuscript figures have been archived. More details can be found
in~\cite{Pal:Zenodo:2024}.


\bibliographystyle{unsrt}
\bibliography{bibliography.bib}

\begin{thebibliography}{10}

\bibitem{Nowak:Science:2006}
M.~A. Nowak.
\newblock Five rules for the evolution of cooperation.
\newblock {\em Science}, 314(5805):1560--1563, 2006.

\bibitem{Wedekind:Science:2000}
C.~Wedekind and M.~Milinski.
\newblock Cooperation through image scoring in humans.
\newblock {\em Science}, 288:850--852, 2000.

\bibitem{Milinski:PRSB:2001}
M.~Milinski, D.~Semmann, T.~C. Bakker, and H.-J. Krambeck.
\newblock Cooperation through indirect reciprocity: image scoring or standing
  strategy?
\newblock {\em Proceedings of the Royal Society B}, 268:2495--2501, 2001.

\bibitem{Seinen:EER:2006}
I.~Seinen and A.~Schram.
\newblock Social status and group norms: {I}ndirect reciprocity in a repeated
  helping experiment.
\newblock {\em European Economic Review}, 50:581--602, 2006.

\bibitem{Engelmann:GEB:2009}
D.~Engelmann and U.~Fischbacher.
\newblock Indirect reciprocity and strategic reputation building in an
  experimental helping game.
\newblock {\em Games and Economic Behavior}, 67:399--407, 2009.

\bibitem{Yoeli:PNAS:2013}
E.~Yoeli, M.~Hoffman, D.~G. Rand, and M.~A. Nowak.
\newblock Powering up with indirect reciprocity in a large-scale field
  experiment.
\newblock {\em Proceedings of the National Academy of Sciences}, 110(Supplement
  2):10424--10429, 2013.

\bibitem{baker2014paying}
W.~E. Baker and N.~Bulkley.
\newblock Paying it forward vs. rewarding reputation: Mechanisms of generalized
  reciprocity.
\newblock {\em Organization Science}, 25(5):1493--1510, 2014.

\bibitem{melamed2020robustness}
D.~Melamed, B.~Simpson, and J.~Abernathy.
\newblock The robustness of reciprocity: Experimental evidence that each form
  of reciprocity is robust to the presence of other forms of reciprocity.
\newblock {\em Science Advances}, 6(23):eaba0504, 2020.

\bibitem{simpson2018roots}
B.~Simpson, A.~Harrell, D.~Melamed, N.~Heiserman, and D.~V. Negraia.
\newblock The roots of reciprocity: Gratitude and reputation in generalized
  exchange systems.
\newblock {\em American Sociological Review}, 83(1):88--110, 2018.

\bibitem{Nowak:Nature:2005}
M.~A. Nowak and K.~Sigmund.
\newblock Evolution of indirect reciprocity.
\newblock {\em Nature}, 437:1291--1298, 2005.

\bibitem{Ohtsuki:JTB:2006a}
H.~Ohtsuki and Y.~Iwasa.
\newblock The leading eight: Social norms that can maintain cooperation by
  indirect reciprocity.
\newblock {\em Journal of Theoretical Biology}, 239:435--444, 2006.

\bibitem{Ohtsuki:JTB:2004a}
H.~Ohtsuki and Y.~Iwasa.
\newblock How should we define goodness? -- {R}eputation dynamics in indirect
  reciprocity.
\newblock {\em Journal of Theoretical Biology}, 231:107--20, 2004.

\bibitem{clark2020indirect}
D.~Clark, D.~Fudenberg, and A.~Wolitzky.
\newblock Indirect reciprocity with simple records.
\newblock {\em Proceedings of the National Academy of Sciences},
  117(21):11344--11349, 2020.

\bibitem{Pfeiffer:PRSB:2005}
T.~Pfeiffer, C.~Rutte, T.~Killingback, M.~Taborsky, and S.~Bonhoeffer.
\newblock Evolution of cooperation by generalized reciprocity.
\newblock {\em Proceedings of the Royal Society B}, 272:1115--1120, 2005.

\bibitem{barta2011cooperation}
Z.~Barta, J.~M. McNamara, D.~B. Huszar, and M.~Taborsky.
\newblock Cooperation among non-relatives evolves by state-dependent
  generalized reciprocity.
\newblock {\em Proceedings of the Royal Society B: Biological Sciences},
  278(1707):843--848, 2011.

\bibitem{Rutte:PLoSB:2007}
C.~Rutte and M.~Taborsky.
\newblock Generalized reciprocity in rats.
\newblock {\em PLoS Biology}, 5:e196, 2007.

\bibitem{Axelrod:Science:1981}
R.~Axelrod and W.~D. Hamilton.
\newblock The evolution of cooperation.
\newblock {\em Science}, 211:1390--1396, 1981.

\bibitem{van2012evolution}
G.~Van~Doorn and M.~Taborsky.
\newblock The evolution of generalized reciprocity on social interaction
  networks.
\newblock {\em Evolution}, 66(3):651--664, 2012.

\bibitem{fudenberg2012slow}
D.~Fudenberg, D.~G. Rand, and A.~Dreber.
\newblock Slow to anger and fast to forgive: Cooperation in an uncertain world.
\newblock {\em American Economic Review}, 102(2):720--749, 2012.

\bibitem{Van-Veelen:PNAS:2012}
M.~van Veelen, J.~Garc\'ia, D.~G. Rand, and M.~A. Nowak.
\newblock Direct reciprocity in structured populations.
\newblock {\em Proceedings of the National Academy of Sciences USA},
  109:9929--9934, 2012.

\bibitem{Hilbe:ncomms:2014}
C.~Hilbe, T.~R\"ohl, and M.~Milinski.
\newblock Extortion subdues human players but is finally punished in the
  prisoner's dilemma.
\newblock {\em Nature Communications}, 5:3976, 2014.

\bibitem{Hilbe:NHB:2018}
C.~Hilbe, K.~Chatterjee, and M.~A. Nowak.
\newblock Partners and rivals in direct reciprocity.
\newblock {\em Nature Human Behaviour}, 2(7):469--477, 2018.

\bibitem{Boyd:SN:1989}
R.~Boyd and P.~J. Richerson.
\newblock The evolution of indirect reciprocity.
\newblock {\em Social Networks}, 11(3):213--236, 1989.

\bibitem{dufwenberg2001direct}
M.~Dufwenberg, U.~Gneezy, W.~G{\"u}th, and E.~Van~Damme.
\newblock Direct vs indirect reciprocity: an experiment.
\newblock {\em Homo Oeconomicus}, 18(1/2):19--30, 2001.

\bibitem{herne2013experimental}
K.~Herne, O.~Lappalainen, and E.~Kestil{\"a}-Kekkonen.
\newblock Experimental comparison of direct, general, and indirect reciprocity.
\newblock {\em The Journal of Socio-Economics}, 45:38--46, 2013.

\bibitem{stanca2009measuring}
L.~Stanca.
\newblock Measuring indirect reciprocity: Whose back do we scratch?
\newblock {\em Journal of Economic Psychology}, 30(2):190--202, 2009.

\bibitem{Molleman:PRSB:2016}
L.~Molleman, E.~van~den Broek, and M.~Egas.
\newblock Personal experience and reputation interact in human decisions to
  help reciprocally.
\newblock {\em Proceedings of the Royal Society B}, 28:20123044, 2013.

\bibitem{Schmid:NHB:2021}
L.~Schmid, K.~Chatterjee, C.~Hilbe, and M.A. Nowak.
\newblock A unified framework of direct and indirect reciprocity.
\newblock {\em Nature Human Behaviour}, 5:1292--1302, 2021.

\bibitem{nakamaru2004evolution}
M.~Nakamaru and M.~Kawata.
\newblock Evolution of rumours that discriminate lying defectors.
\newblock {\em Evolutionary Ecology Research}, 6(2):261--283, 2004.

\bibitem{seki2016model}
M.~Seki and M.~Nakamaru.
\newblock A model for gossip-mediated evolution of altruism with various types
  of false information by speakers and assessment by listeners.
\newblock {\em Journal of Theoretical Biology}, 407:90--105, 2016.

\bibitem{rankin2009assortment}
D.~J. Rankin and M.~Taborsky.
\newblock Assortment and the evolution of generalized reciprocity.
\newblock {\em Evolution}, 63(7):1913--1922, 2009.

\bibitem{nowak2007upstream}
M.~A. Nowak and S.~Roch.
\newblock Upstream reciprocity and the evolution of gratitude.
\newblock {\em Proceedings of the Royal society B: Biological Sciences},
  274(1610):605--610, 2007.

\bibitem{sasaki2023evolution}
T.~Sasaki, S.~Uchida, I.~Okada, and H.~Yamamoto.
\newblock The evolution of cooperation and diversity by integrated indirect
  reciprocity.
\newblock {\em arXiv preprint arXiv:2303.04467}, 2023.

\bibitem{Berger:GEB:2011}
U.~Berger.
\newblock Learning to cooperate via indirect reciprocity.
\newblock {\em Games and Economic Behavior}, 72:30--37, 2011.

\bibitem{Sigmund:book:2010}
K.~Sigmund.
\newblock {\em The Calculus of Selfishness}.
\newblock Princeton Univ. Press, Princeton, NJ, 2010.

\bibitem{Imhof:PRSB:2010}
L.~A. Imhof and M.~A. Nowak.
\newblock Stochastic evolutionary dynamics of direct reciprocity.
\newblock {\em Proceedings of the Royal Society B}, 277:463--468, 2010.

\bibitem{baek2016comparing}
S.~K. Baek, H-C. Jeong, C.~Hilbe, and M.~A Nowak.
\newblock Comparing reactive and memory-one strategies of direct reciprocity.
\newblock {\em Scientific Reports}, 6(1):25676, 2016.

\bibitem{rockenbach2006efficient}
B.~Rockenbach and M.~Milinski.
\newblock The efficient interaction of indirect reciprocity and costly
  punishment.
\newblock {\em Nature}, 444(7120):718--723, 2006.

\bibitem{Okada:SciRep:2018}
I.~Okada, H.~Yamamoto, Y.~Sato, S.~Uchida, and T.~Sasaki.
\newblock Experimental evidence of selective inattention in reputation-based
  cooperation.
\newblock {\em Scientific Reports}, 8:14813, 2018.

\bibitem{Sommerfeld:PNAS:2007}
R.~D. Sommerfeld, H.-J. Krambeck, D.~Semmann, and M.~Milinski.
\newblock Gossip as an alternative for direct observation in games of indirect
  reciprocity.
\newblock {\em Proceedings of the National Academy of Sciences USA},
  104(44):17435--17440, 2007.

\bibitem{tkadlec2023mutation}
Josef Tkadlec, Christian Hilbe, and Martin~A Nowak.
\newblock Mutation enhances cooperation in direct reciprocity.
\newblock {\em Proceedings of the National Academy of Sciences},
  120(20):e2221080120, 2023.

\bibitem{hamilton2005contingent}
I.~M. Hamilton and M.~Taborsky.
\newblock Contingent movement and cooperation evolve under generalized
  reciprocity.
\newblock {\em Proceedings of the Royal Society B: Biological Sciences},
  272(1578):2259--2267, 2005.

\bibitem{ohtsuki2007direct}
H.~Ohtsuki and M.~A. Nowak.
\newblock Direct reciprocity on graphs.
\newblock {\em Journal of Theoretical Biology}, 247(3):462--470, 2007.

\bibitem{Pacheco:JTB:2008}
J.~M. Pacheco, A.~Traulsen, H.~Ohtsuki, and M.~A. Nowak.
\newblock Repeated games and direct reciprocity under active linking.
\newblock {\em Journal of Theoretical Biology}, 250:723--731, 2008.

\bibitem{sigmund2001reward}
K.~Sigmund, C.~Hauert, and M.~A. Nowak.
\newblock Reward and punishment.
\newblock {\em Proceedings of the National Academy of Sciences},
  98(19):10757--10762, 2001.

\bibitem{Milinski:PNAS:1998}
M.~Milinski and C.~Wedekind.
\newblock Working memory constrains human cooperation in the prisoner's
  dilemma.
\newblock {\em Proceedings of the National Academy of Sciences USA},
  95:13755--13758, 1998.

\bibitem{duffy2014cognitive}
S.~Duffy and J.~Smith.
\newblock Cognitive load in the multi-player prisoner's dilemma game: Are there
  brains in games?
\newblock {\em Journal of Behavioral and Experimental Economics}, 51:47--56,
  2014.

\bibitem{Santos:Nature:2018}
F.~P. Santos, F.~C. Santos, and J.~M. Pacheco.
\newblock Social norm complexity and past reputations in the evolution of
  cooperation.
\newblock {\em Nature}, 555:242--245, 2018.

\bibitem{watanabe2014two}
T.~Watanabe, M.~Takezawa, Y.~Nakawake, A.~Kunimatsu, H.~Yamasue, M.~Nakamura,
  Y.~Miyashita, and N.~Masuda.
\newblock Two distinct neural mechanisms underlying indirect reciprocity.
\newblock {\em Proceedings of the National Academy of Sciences},
  111(11):3990--3995, 2014.

\bibitem{traulsen2006stochastic}
A.~Traulsen, M.~A. Nowak, and J.~M. Pacheco.
\newblock Stochastic dynamics of invasion and fixation.
\newblock {\em Physical Review E}, 74(1):011909, 2006.

\bibitem{Press:PNAS:2012}
W.~H. Press and F.~D. Dyson.
\newblock Iterated prisoner's dilemma contains strategies that dominate any
  evolutionary opponent.
\newblock {\em Proceedings of the National Academy of Sciences},
  109:10409--10413, 2012.

\bibitem{Pal:Zenodo:2024}
S.~Pal, C.~Hilbe, and N.~E. Glynatsi.
\newblock {[DATA] The co-evolution of direct, indirect and generalized
  reciprocity}, 2024.
\newblock Available online at: \url{https://doi.org/10.5281/zenodo.14035109}.

\end{thebibliography}


\begin{thebibliography}{1}

\bibitem{Schmid:NHB:2021}
L.~Schmid, K.~Chatterjee, C.~Hilbe, and M.A. Nowak.
\newblock A unified framework of direct and indirect reciprocity.
\newblock {\em Nature Human Behaviour}, 5:1292--1302, 2021.

\bibitem{Press:PNAS:2012}
W.~H. Press and F.~D. Dyson.
\newblock Iterated prisoner's dilemma contains strategies that dominate any
  evolutionary opponent.
\newblock {\em Proceedings of the National Academy of Sciences},
  109:10409--10413, 2012.

\bibitem{Imhof:PRSB:2010}
L.~A. Imhof and M.~A. Nowak.
\newblock Stochastic evolutionary dynamics of direct reciprocity.
\newblock {\em Proceedings of the Royal Society B}, 277:463--468, 2010.

\bibitem{fudenberg:JET:2006}
D.~Fudenberg and L.~A. Imhof.
\newblock Imitation processes with small mutations.
\newblock {\em Journal of Economic Theory}, 131:251--262, 2006.

\end{thebibliography}
\end{document}


\maketitle
\onehalfspacing
\section{Model}
\label{Section:model-repeated-games}
\subsection*{Setup of repeated games in a finite population}
We consider a finite population of $n$ individuals (called players). In this population, players only have pairwise interactions with each other. At each (discrete) time point, a single pair interacts in the population. After a pairwise interaction, the next interaction in the population occurs with probability $d$. In an interaction, two players are randomly drawn from the $n$ players to play the donation game. In the donation game, a player can either cooperate, wherein they give a benefit of $b$ to their co-player while incurring a cost, $c< b$, or defect, in which case they do nothing. Their actions are observable to everyone in the population. 
\\

\noindent In our model, players can make errors in observing actions in interactions in which they are not taking part (we call them 3$^{rd}$ party interactions). Each player independently mistakes cooperation for defection (and \emph{vice versa}) with probability $\varepsilon \in [0,1/2]$ in the model. Players do not make an observation error if they are part of the interaction. 
\\

\noindent We define the pairwise continuation probability, $\delta$, as the conditional probability that a pair of players will meet again in the future if they meet in the present. This probability $\delta$, can be derived in terms of the population continuation probability, $d$ and population size, $n$ as~\cite{Schmid:NHB:2021}, \\
\begin{equation}
\delta = \frac{2d}{2d + n(n-1)(1-d)}.
\label{Eq:delta-with-d-n}
\end{equation}

\subsection*{Strategies of players}

In the model, players keep track of others' reputation. To every player in the population, everyone other player is in a binary standing -- they are either \textit{good} or \textit{bad}. A player cooperates with their co-player in an interaction if and only if they are in a \textit{good} standing (otherwise they defect).\\

\noindent  After an interaction, all players have the chance to update reputations of others by using their strategy and the (possibly erroneous) observation from the interaction. The strategy of a player $i$, $\sigma_i$, is a tuple with five elements: $\sigma_i = (y_i,p_i,q_i,\lambda_i,\gamma_i)$. Each element of the strategy is a probability and is therefore in the interval $[0,1]$. We describe the strategy components below:
\begin{enumerate}

\item $y_i$ (initialization): The probability that $i$ assigns a \textit{good} reputation to a player $j$ when a) $i$ has no reputational history of $j$ (i.e., she has never updated $j$'s reputation before) 

\item $p_i$ (reciprocal cooperativeness): The probability that $i$ assigns a \textit{good} reputation to an individual if she observes a cooperation by them in direct interaction.

\item $q_i$ (generosity): The probability that $i$ assigns a \textit{good} reputation to an individual if she observes a defection by them in direct interaction.

\item $\lambda_i$ (imporatance given to third party interactions, or \textit{indirect reciprocity}): The probability that $i$ updates the reputation of player $j$ after an interaction that involved $j$ and $k$. While updating, $i$ assumes that $j$ plays against her what she observes him playing against $k$. That is, she uses the probabilities $p_i$, $q_i$ and the observed action of $j$ against $k$ to update $j$'s reputation. 

\item $\gamma_i$ (generalizing the action of an interaction partner to the whole population or, \textit{generalized reciprocity}): The probability that $i$ updates the reputation of player $j$ after an interaction that involved $i$ and $k$ but not $j$. After her interaction with $k$, player $i$, chooses to update the reputation of all $j \neq k$ in the population independently with probability $\gamma_i$. While updating the reputation of a particular $j$, $i$ assumes that $j$ plays the same action that $k$ played against her in the interaction.

\end{enumerate}

\noindent If player $i$ uses $\lambda_i = 0$ in her strategy, she does not update the reputation of others if she, herself, was not part of the respective interaction. In other words, she only update reputations based on what was done to her (and ignore what others do between each other). If she uses $\gamma_i = 0$, she does not update the reputation of others based on what someone else did to her. We assume players always update the reputation of an individual after directly interacting with them. 

\subsection*{Recursion of reputations in an arbitrary population composition}

We denote with $x_{ij}(t)$ the probability that a player $i$ considers $j$ \textit{good} after $t$ interactions. This probability at time $t+1$, $x_{ij}(t+1)$, can be written using values of $x$ at time $t$ using the following recursive relation,
\begin{equation}
\begin{split}
x_{ij}(t + 1) &= (1 - 2\bar{w} + w)x_{ij}(t)  \\[10pt]
&+ w(x_{ji}(t)p_i + (1 - x_{ji}(t))q_i) \\[10pt]
&+(\bar{w} - w)(1 - \lambda_i)x_{ij}(t) \\[10pt]
&+ w\lambda_i \mathlarger{\sum}_{l \neq i,j} (\bar{\varepsilon}x_{jl}(t) + \varepsilon \bar{x}_{jl}(t))p_i + (\bar{\varepsilon} \bar{x}_{jl}(t) + \varepsilon x_{jl}(t))q_i \\[10pt]
&+(\bar{w} - w)(1 - \gamma_i)x_{ij}(t) \\[10pt]
&+ w \gamma_i \mathlarger{\sum}_{l \neq i,j} x_{li}(t)p_i + \bar{x}_{li}(t)q_i \\[12pt]
\end{split}
\label{Eq:recursion_1}
\end{equation}
Where $w = 2/(n(n-1))$ is the probability that a particular pair of individuals is selected for an interaction and $\bar{w} = 2/n$ is the probability that a particular player is selected in an interaction. We use the shorthand notation $\bar{\varepsilon} := 1 - \varepsilon$ to denote the probability that no error was made while observing 3$^{rd}$ party interactions. We also use the shorthand notation $\bar{x}$ to denote $1-x$. The above recursion covers all the cases in which $i$ may update $j$'s reputation after the interaction in round $t+1$. Each of these cases corresponds to a line in Eq.~(\ref{Eq:recursion_1}). 

\begin{enumerate}
    \item In the first line, we cover the case where neither $i$ nor $j$ is selected for an interaction in time $t+1$. This happens with probability $1 - 2\bar{w} + w$. In this case, $i$ makes no change to $j$'s reputation. 
    \item In the second line, we cover the case where both $i$ and $j$ were selected for an interaction in round~$t$. This happens with probability $w$. In this interaction, $j$ cooperates with probability $x_{ji}(t)$ and defects with with probability $1-x_{ji}(t)$. If $j$ cooperates, he is updated to \textit{good} by $i$ with probability~$p_i$. If he defects, he is updated to \textit{good} with probability $q_i$.
    
    \item In the third line, we cover the case where $j$ was part of an interaction but not $i$ (this happens with probability $\bar{w} - w$), \emph{and} $i$ chooses to ignore $j$'s interaction (with probability $1-\lambda_i$), therefore there is no update to $j$'s reputation in round $t+1$.
    
    \item In the fourth line, we cover all the cases where $j$ interacts with an arbitrary third party $l$ (with probability $w$) and ~$i$ takes this interaction into account (with probability $\lambda_i$). If $j$ appears to cooperate (``appears" because ~$i$'s observation is subject to errors), he is assigned \textit{good} by $i$ with probability ~$p_i$. Otherwise, if they appear to defect, they are assigned a \textit{good} reputation with probability $q_i$.
    
    \item In the fifth line, we cover the case where ~$i$ but not ~$j$ was part of an interaction (with probability ~$\bar{w} - w$) but ~$i$ chose to not update ~$j$'s reputation (with probability ~$1-\gamma_i$). 
    
    \item In the sixth line, we cover all the cases where $i$ has an interaction with third party $l$ (with probability $w$) and she chooses to update $j$'s reputation based on this interaction. If $l$ cooperates with~$i$, then $j$ is assigned a \textit{good} reputation with probability $p_i$ and if $l$ defects, $j$ is assigned a \textit{good} reputation with probability $q_i$.  \\[5pt]
\end{enumerate}

\noindent By denoting $p_i - q_i := r_i$, one can simplify the recursion in Eq.~(\ref{Eq:recursion_1}) as, 

\begin{equation}
\begin{split}
x_{ij}(t + 1) = &(1 - w - (\bar{w} - w)(\lambda_i + \gamma_i))x_{ij}(t)  \\[10pt]
&+ w r_i x_{ji}(t)  \\[10pt]
&+ w\lambda_i r_i (1 - 2 \varepsilon) \sum_{l \neq i,j} x_{jl}(t) \\[10pt]
&+ w \gamma_i r_i \sum_{l \neq i,j} x_{li}(t)  \\[10pt]
&+ wq_i + w \gamma_i q_i (n-2) + w \lambda_i (q_i + \varepsilon r_i) (n-2) \\[10pt]
\end{split}
\label{Eq: recursion_2}
\end{equation}
\clearpage
\section{Computation of expected payoffs}
\label{Section:computation-exact-payoffs-repeated}

\noindent We collect all the probabilities $x$ in a column vector of a size $n(n-1)$, 
\begin{equation}
\xbf(t) := (x_{12}(t),x_{13}(t),\cdots,x_{1n}(t); \cdots ; x_{n1}(t),x_{n2}(t),\cdots,x_{n(n-1)}(t))^\intercal
\label{Eq:vector-x}
\end{equation}
\noindent and define the square matrix $\Mbf := (m_{ij,kl})$ and the column vector $\vbf := (v_{ij})$ as the following, \\[2pt]
\begin{equation}
m_{ij,kl} = \begin{cases}
1 - w - (\bar{w} - w)(\lambda_i + \gamma_i) \quad &\text{ if } i = k, j = l \\[10pt]
wr_i \quad &\text{ if } i = l, j = k \\[10pt]
w\lambda_i r_i (1 - 2 \varepsilon) &\text{ if } j = k, \text{ and } i,j \neq l \\[10pt]
w\gamma_ir_i  &\text{ if } i = l \text{ and } i,j \neq k \\[10pt]
0 &\text{ otherwise } \\[10pt]
\end{cases}
\label{Eq:matrix-M}
\end{equation}
\noindent and, \\ 
\begin{equation}
v_{ij} = wq_i + w \gamma_i q_i (n-2) + w \lambda_i (q_i + \varepsilon r_i) (n-2). \\[10pt]
\label{Eq:vector-v}
\end{equation}
\noindent Then, the recursion in Eq.~(\ref{Eq: recursion_2}) takes the following vector form,
\begin{equation}
\xbf(t+1) = \Mbf \xbf(t) + \vbf
\label{Eq:recursion-vector}
\end{equation}
\noindent where $\xbf(0)$ is just the collection of the first elements, $y$, from players' strategies. That is, $x_{ik}(0) = y_i$ for all $i$ and $k$. We can compute the average probability that players consider each other \textit{good}, $\xbf$, by time discounting the vector $\xbf(t)$ in the following way, 
\begin{equation}
    \xbf := (1-d) \sum_{t = 0}^\infty d^t \xbf(t). 
    \label{Eq:discounted_x}
\end{equation}
\noindent Now, by using the recursive relation in Eq.~(\ref{Eq:recursion-vector}) and Eq~(\ref{Eq:discounted_x}), one can derive the following expression for the average probability vector, $\xbf$,
\begin{equation}
    \xbf = (\mathbb{1}-d\Mbf)^{-1}((1-d)\xbf_0 + d\vbf)
\end{equation}
\noindent where $\xbf_0 = \xbf(0)$. One can interpret elements of $\xbf$ as the probability that players consider each other \textit{good} (and therefore cooperate with others) in a randomly selected round. The expected payoff of player $i$ is therefore,
\begin{equation}
\pi_i = \frac{1}{n-1} \sum_{j \neq i} x_{ji}b - x_{ij}c. \\[10pt]
\label{Eq:payoff-equation}
\end{equation}
\section{The Nash equilibria of the game}
In this section we characterize all Nash equilibria strategies of the game with the form $(y,p,q,\lambda,\gamma)$. The strategy $\sigma := (y,p,q,\lambda,\gamma)$ is a Nash equilibrium strategy if no unilateral deviation by a single player in a homogeneous population playing $\sigma$ leads the deviating player to a strictly better (expected) payoff. In our characterization of Nash strategies, we allow deviating players to deviate from $\sigma$ to \emph{any} arbitrary strategy, even strategies that cannot be represented by a 5-tuple. We first present some auxiliary results.

\begin{Prop} \label{Prop:coop-payoff-homogenous}
(Cooperation and payoffs in homogenous population) Consider a population of size $n$ where every individual plays $(y,p,q,\lambda,\gamma)$ in a game where $\delta < 1$. Let $r := p - q$. Then, the average probability that players cooperate with each other is \\[10pt]
\begin{equation}
x = \frac{(1 - \delta)y + \delta q + \delta (n-2)(q(\lambda + \gamma) + \lambda \varepsilon r)}{1 - \delta r + \delta (n-2)((\lambda + \gamma)(1 - r) + 2 \lambda \varepsilon r)}. \\[10pt]
\label{Eq:coop-rate-homogenous}
\end{equation}
\noindent In this homogenous population each player receives an average payoff, $\pi = (b-c)x$. \\ \\ 
Further, if the game is generic, i.e, $n > 2, \varepsilon > 0$, a population is 
\begin{enumerate}
    \item fully defective ($x = 0$) if and only if, $y = 0$, $q = 0$ and, in addition, $\lambda = 0$ or $p = 0$.
    \item fully cooperative ($x = 1$) if and only if, $y = 1, p =1$ and, in addition, $\lambda = 0$ or $q = 1$. 
\end{enumerate}
If the game is non-generic, i.e., $n = 2$ or $\varepsilon = 0$, a population is,
\begin{enumerate}
    \item fully defective if and only if $y = q = 0$. 
    \item fully cooperative if and only if $y = p = 1$. 
\end{enumerate}
\end{Prop}
\noindent All proofs of propositions, lemmas and theorems are in the appendix.
\\ \\ 
\noindent Now consider a population where players $1,2,...,n-1$ are playing the same strategy $\sigma = (y,p,q,\lambda,\gamma)$ while player $n$ is playing an arbitrary strategy. For simplicity, we call these $n-1$ players as \emph{residents} and the last player, $n$ as \emph{mutant}. In the following lemmas we derive some linear relationships between the average rates of cooperation in this population, 

\begin{lemma}
Consider $n > 2$, and the game with pairwise continuation probability, $\delta$, and observation error probability, $\varepsilon$. Let $i,k \in \{1,2,...,n-1\}$ such that $i \neq k$. All players in $\{1,2,...,n-1\}$ play the strategy $(y,p,q,\lambda,\gamma)$. Then, the following linear relationships are satisfied, independent of the arbitrary strategy that player $n$ adopts,
\begin{equation}
\begin{split}
x_{in} &= B_1 + B_2 x_{ni} + B_3 x_{ik} \\[10pt]
x_{ik} &= C_1 + C_2 x_{ni} + C_3 x_{in}
\end{split}.
\label{Eq: first-linear-equation}
\end{equation} \\
\noindent The coefficients $B_1$, $B_2$, $B_3$ and $C_1$, $C_2$, $C_3$ are given by\\[10pt]
\begin{equation}
\begin{split}
B_1  &= \frac{(1 - \delta)y + \delta q + \delta (n-2)(q(\lambda + \gamma) + \lambda \varepsilon r)}{1 + \delta(n-2)(\lambda + \gamma)} \\[10pt]
B_2  &= \frac{\delta r (1  + \lambda (n-2)(1 - 2 \varepsilon))}{1 + \delta(n-2)(\lambda + \gamma)} \\[10pt]
B_3  &= \frac{\delta \gamma r (n-2)}{1 + \delta(n-2)(\lambda + \gamma)} \\[15pt]
C_1 
&= \frac{(1 - \delta)y + \delta q + \delta (n - 2)(q(\lambda + \gamma) + \lambda \varepsilon r)}{1 - \delta r + \delta (\lambda + \gamma)(n-2) - \delta r (n-3)(\lambda + \gamma - 2 \lambda \varepsilon)
} \\[10pt]
C_2 
&= \frac{\delta \gamma r}{1 - \delta r + \delta (\lambda + \gamma)(n-2) - \delta r (n-3)(\lambda + \gamma - 2 \lambda \varepsilon)
)} \\[10pt]
C_3
&= \frac{\delta \lambda r (1 - 2 \varepsilon)}{1 - \delta r + \delta (\lambda + \gamma)(n-2) - \delta r (n-3)(\lambda + \gamma - 2 \lambda \varepsilon)
}
\end{split}
\label{Eq:coefficients-Bs}
\end{equation}
\label{claim: first linear relationship}
\end{lemma}
\noindent The average probability that a resident cooperates with player $n$, $x_{in}$, varies linearly with the average probability that residents cooperate with residents $x_{ik}$, and player $n$ cooperates with residents, $x_{ni}$. 
Using Lemma~\ref{claim: first linear relationship} we can show that a linear relationship exists between $x_{in}$ and $x_{ni}$ for any strategy that the mutant adopts,

\begin{Prop}
Consider the game with $n$ players, pairwise continuation probability, $\delta$, and observation error probability $\varepsilon$. All players in set $\{1,2,...,n-1\}$ play the strategy $(y,p,q,\lambda,\gamma)$ and player $n$ plays an arbitrary strategy. Then, the following relationship between average cooperation shown by residents to the mutant, $x_{in}$, and the average cooperation shown by the mutant to the residents, $x_{ni}$, is satisfied irrespective of the strategy the mutant (player $n$) adopts,
\begin{equation}
x_{in} = K_1 + K_2x_{ni}.
\label{Eq:linear-eq-more-than-2}
\end{equation}
The coefficients $K_1$ and $K_2$ can be expressed in terms of the coefficients given in Eq.~(\ref{Eq:coefficients-Bs}) as,
\begin{equation}
K_1 = 
\begin{cases}
\displaystyle \frac{B_1 + B_3 C_1}{1 - B_3 C_3} &\quad \text{ if } n > 2 \\[10pt]
(1-\delta)y + \delta q &\quad \text{ if } n = 2 
\end{cases}\\[20pt]
\label{Eq:coefficients-K1}.
\end{equation}
\begin{equation}
K_2 = 
\begin{cases}
\displaystyle \frac{B_2 + B_3 C_2}{1 - B_3 C_3} &\quad \text{ if } n > 2 \\[10pt]
\delta r &\quad \text{ if } n = 2
\end{cases}
\label{Eq:coefficients-K2}.
\end{equation}
As a result, the mutant's payoff can be written as,
\begin{equation}
\pi_n = K_1b + (K_2b - c)x_{ni}
\label{Eq:payoff-mutant-linear}
\end{equation}
\label{Prop:average-payoff-mutant-linear}
\end{Prop}
\noindent The proposition above demonstrates that the resident strategy enforces a linear relationship on the payoffs of the mutant to how much the mutant cooperates with them (this is similar to zero-determinant strategies in the literature of direct reciprocity~\cite{Press:PNAS:2012}). When the slope, $K_2b - c < 0$ is set to be strictly negative by the resident strategy, the best response for the mutant is to always defect. Similarly, when the slope is set to be strictly positive, the best response is to always cooperate. When the slope is exactly zero, any strategy that the mutant adopts is a best response as all strategies earn the same payoff of $K_1b$. In this case, the resident strategy is also a best-response to itself. While characterizing all Nash equilibria strategies of the form $(y,p,q,\lambda,\gamma)$, we study the sign of the slope, $K_2b - c$, of this linear relationship. In the next lemma, we compare the signs of the slope $K_2b - c$ and a quadratic in $r := p - q$. This simplifies the characterization of Nash strategies in the next section. 

\begin{lemma}
\label{claim:algebraic-relationship}
Consider $n > 2$. Then, if we define the coefficients $\Lambda_1, \Lambda_2, \Lambda_3$ of the quadratic in $r$, $\Lambda_1 r^2 + \Lambda_2 r - \Lambda_3$ in the following way, 


\begin{equation}
    \begin{split}
        \Lambda_1 &=  \delta^2 \gamma (n-2)(\gamma b + \lambda c (1 - 2 \varepsilon)) - b \delta D (1 + \lambda(n-2)(1 - 2\varepsilon))\\
        \Lambda_2 &= b \delta A (1 + \lambda (n-2)(1 - 2 \varepsilon)) + cAD  
        \\
        \Lambda_3 &= cA^2
    \end{split}
\label{Eq:coefficients-Lambda}
\end{equation}

\noindent where $A$ and $D$ are given by, 
\begin{equation}
    \begin{split}
        A &= 1 + \delta (n-2)(\lambda + \gamma) \\[10pt]
        D &= \delta (1 + (n-3)(\lambda + \gamma - 2\lambda \varepsilon))
    \end{split}
\label{Eq:coefficients-AD}
\end{equation}
then, the following statements are true, 
\begin{enumerate}
    \item $K_2b - c = 0 \Leftrightarrow \Lambda_1 r^2 + \Lambda_2 r - \Lambda_3 = 0$
    \item $K_2b - c > 0 \Leftrightarrow \Lambda_1 r^2 + \Lambda_2 r - \Lambda_3 > 0$ 
    \item $K_2b - c < 0 \Leftrightarrow \Lambda_1 r^2 + \Lambda_2 r - \Lambda_3 < 0$
\end{enumerate}
\noindent where $K_2$ is given by Eq.~(\ref{Eq:coefficients-K2}). 
\end{lemma}
\noindent The above lemma shows that the signs of $K_2b - c$ and the quadratic in $r$, $\Lambda_1 r^2 + \Lambda_2 r - \Lambda_3$ are the same when $n>2$. The result from this lemma allows us to focus only on the properties of the quadratic in $r$ instead of studying the properties of $K_2b - c$ while characterizing the Nash equilibria. 

\subsection*{Characterization of Nash equilibria}

\noindent In the two following theorems, we characterize all Nash equilibria of the form $(y,p,q,\lambda, \gamma)$ for games where $\delta < 1$.
\begin{theorem}
(characterization of Nash equilibria for $n > 2$) \\ For $n >2$ and $\delta < 1$, a strategy $(y,p,q,\lambda,\gamma)$ is a Nash equilibrium if and only if atleast one of the following conditions holds,
\begin{enumerate}
    \item The strategy always defects. That is, $y = p = q = 0$ (a generic Nash)
    \item It is an equalizer strategy that sets a mutant's payoff to a constant value. \\That is, $\Lambda_1 r^2 + \Lambda_2 r - \Lambda_3 = 0$ (a generic Nash)
    \item $y = p = 1$ \textbf{and} $\Lambda_1 (1-q)^2 + \Lambda_2 (1-q) - \Lambda_3 > 0$ \textbf{and} $\varepsilon = 0$ or $\lambda = 0$
    \item $y = q = 0$ \textbf{and} $\Lambda_1 p^2 + \Lambda_2 p - \Lambda_3 < 0$ \textbf{and} $\varepsilon = 0$ or $\lambda = 0$.
\end{enumerate}
\noindent where $r := p - q$ and $\Lambda_1$, $\Lambda_2$ and $\Lambda_3$ are given by Eq.~(\ref{Eq:coefficients-Lambda}) and Eq.~(\ref{Eq:coefficients-AD}).
\label{Th:generic-Nash-equilibria}
\end{theorem}

\noindent We call the set containing \textrm{ALLD} and all equalizer strategies (the strategies that set mutant's payoff to a constant value), the set of \textit{generic} Nash equilibria.  
\begin{theorem}
(characterization of Nash equilibria for $n = 2$) \\For $n=2$ and $\delta < 1$, a strategy $(y,p,q,\lambda,\gamma)$ is a Nash equilibrium of the game if and only if one of the following is true:
\begin{enumerate}
    \item $p-q =c/(b\delta)$,
    \item $y = q = 0$ and $p < c/(b\delta)$
    \item $y = p = 1$ and $q < 1 - c/(b\delta)$
\end{enumerate}
\label{Th:two-player-Nash-equilibria}
\end{theorem}
\noindent The theorem above states the Nash equilibria for games with only two players in the population. In this case, by definition, indirect and generalized reciprocity are absent.

\subsection*{Cooperative Nash equilibria}
\noindent The theorems in the previous subsection fully characterize the set of all Nash equilibria. In this subsection, we study the subset of all Nash equilibria which sustain full cooperation in the population in the limit where errors are absent.

\begin{definition} (generic cooperative Nash equilibrium):
A Nash equilibrium strategy, $\sigma$, is a (generic) cooperative Nash equilibrium when it is 
\begin{enumerate}
    \item a generic Nash equilibrium and, 
    \item in a homogenous population playing $\sigma$, the cooperation rate according to Eq.~(\ref{Eq:coop-rate-homogenous}) approaches one in the limit of rare errors. That is $x \to 1$ as $\varepsilon \to 0$.
\end{enumerate}
\end{definition}

\noindent The definition above distinguishes between a cooperative Nash equilibrium and a \emph{generic} cooperative Nash equilibrium. Generic cooperative Nash equilibria are a subset of cooperative Nash equilibrium. They require that the Nash equilibrium strategies are fully cooperative between themselves and also equalize the payoff of any mutant. 

\begin{Prop} (characterization of generic cooperative Nash equilibria) \\
For a game where $\delta < 1$, a strategy $(y,p,q,\lambda,\gamma)$ is a generic cooperative Nash equilibrium if and only if 
\begin{enumerate}
    \item $y = p = 1$, and, 
    \item $\Lambda_1 (1-q)^2 + \Lambda_2 (1-q) - \Lambda_3 = 0$
\end{enumerate}
\label{Prop:cooperative-Nash}
\end{Prop}

\noindent The conditions for a strategy with only direct and indirect reciprocity to be a cooperative Nash equilibria have been studied in detail in the original model from Schmid et al.~\cite{Schmid:NHB:2021}. In the following propositions, we present properties of generic cooperative Nash equilibrium with direct (DR) and generalized reciprocity (GR) only ($\lambda = 0$). In these propositions, we naturally assume $n > 2$ since the conditions involve generalized reciprocity. 

\begin{Prop}
(Existence and uniqueness of a generic cooperative Nash with pure GR) 
\\A strategy of the form $(1,1,q^{*}_{\mathrm{G}},0,1)$ is a cooperative Nash equilibrium if and only if $\delta > c/b$ and, 
\begin{equation}
q^{*}_{\mathrm{G}} = 1 - \frac{c + c\delta(n-2)}{b\delta + c\delta(n-2)}.
\label{eq:pure-GR-coop}
\end{equation}
\label{Prop:GR-coop-Nash}
\end{Prop}
\noindent Interestingly, the condition $\delta > c/b$ is the same necessary condition as in the case of direct reciprocity. However, for a fixed value of $\delta > c/b$, the generosity, $q^*_{\mathrm{D}}$, at a pure DR equilibrium is strictly greater than the generosity, $q^*_{\mathrm{G}}$, at a pure GR equilibrium,
\begin{equation}
    1 - \frac{c}{b\delta} = q^*_{\mathrm{D}} > q^*_{\mathrm{G}} \\[5pt].
\end{equation}

\begin{Prop}
(Existence of a unique generic cooperative Nash with only DR and GR) \\ Consider a $\gamma \in [0,1]$. Then, if and only if the following implicit condition is satisfied,
\begin{equation}
\delta(1 - \gamma)((n-2)\gamma + 1) \leq (1 + \delta (n-2) \gamma)\left( 1 + \frac{c}{b} - \frac{c \gamma}{b} - \frac{c}{b \delta} \right) \\[10pt]
\label{Eq:delta-min-GR}
\end{equation}
then there is a unique $q^* \in [0,1)$ such that the strategy  $(1,1,q^*,0,\gamma)$ is a generic cooperative Nash equilibrium.
\label{Prop:Existence-Uniqueness-GR-equilibira}
\end{Prop}

\noindent The above proposition lays down a necessary and sufficient condition on the pairwise continuation probability, $\delta$, for the existence of a unique generic cooperative Nash equilibria with a specific mixture of direct and generalized reciprocity (determined by the strategy element $\gamma$) and no indirect reciprocity (strategy element $\lambda = 0$). By plugging in $\gamma = 0$ and $\gamma = 1$ in the implicit condition above, one recovers identical necessary conditions, $\delta > c/b$. \\ 
\begin{Prop}
\noindent Consider a $\gamma \in [0,1]$. If $(1,1,q^*,0,\gamma)$ is the unique generic cooperative Nash equilibrium of the repeated game, then for any $q \in [0,q^*)$, $(1,1,q,0,\gamma)$ is a (non-generic) cooperative Nash equilibrium. Furthermore, there is no $q > q^*$ such that $(1,1,q,0,\gamma)$ is a cooperative Nash equilibrium. 
\label{Prop:highest-generosity-among-coop-NE}
\end{Prop}

\noindent This proposition is about the set of cooperative Nash equilibrium strategies that do not use indirect reciprocity ($\lambda = 0$) and only use a particular mixture of direct and generalized reciprocity (determined by their shared $\gamma$ value). It claims that if this set contains a unique generic Nash equilibrium, then, the only other element in this set is a continuum of non-generic Nash equilibrium. In addition, all Nash equilibrium in this continuum display a generosity $q$ strictly less than the generosity of the generic Nash equilibrium $q^*$.\\

\noindent \textbf{Q: What is the minimum pairwise continuation probability necessary for a combination of direct and generalized reciprocity to support a cooperative equilibrium?} \\

\noindent A: The minimum continuation probability needed for direct or generalized reciprocity to individually support a cooperative equilibrium ($c/b$) is too low for a combination of them to do the same. We find that the minimum continuation probability necessary for a DR-GR mixture to form a Nash equilibrium depends on how much weight is put on each reciprocity (i.e., the value of $\gamma$). More precisely, with numerical calculations, we find that the condition in Eq. (\ref{Eq:delta-min-GR}) must be satisfied by the pairwise continuation probability $\delta$. \\

\noindent It was observed in Schmid et al. that if both direct and indirect reciprocity are able to independently constitute a cooperative Nash equilibrium, any possible mixture of them can do the same~\cite{Schmid:NHB:2021}. Interestingly, an equivalent statement involving direct and generalized reciprocity does not hold. 
We show that for certain $\delta > c/b$, it is possible that direct and generalized reciprocity are able to independently stabilize cooperation but a strategy that uses a mixture of them cannot. 
To elucidate this result, we perform a deviation analysis (Supplementary Fig.~\ref{fig:GR-equilibrium}). We imagine a single individual deviating to pure defection (the mutant) from an otherwise homogenous resident population. In the resident population, everyone plays the strategy $(1,1,0.01,0,\gamma)$. When residents only use direct reciprocity ($\gamma = 0$), they cooperate with other residents but rarely forgive the defector. As a result the mutant's payoff is below $b-c$, the payoff of the homogenous resident population. In contrast, when residents only use generalized reciprocity ($\gamma = 1$), cooperative residents treat everyone equally but they almost always defect. Also here the mutant receives less than $b-c$. However, when residents use a mixture of the two kinds of reciprocity ($0<\gamma<1$), switching to defection can be strictly beneficial (Supplementary Fig.~\ref{fig:GR-equilibrium}B). 

\begin{figure}
    \centering
    \includegraphics[width = \textwidth]{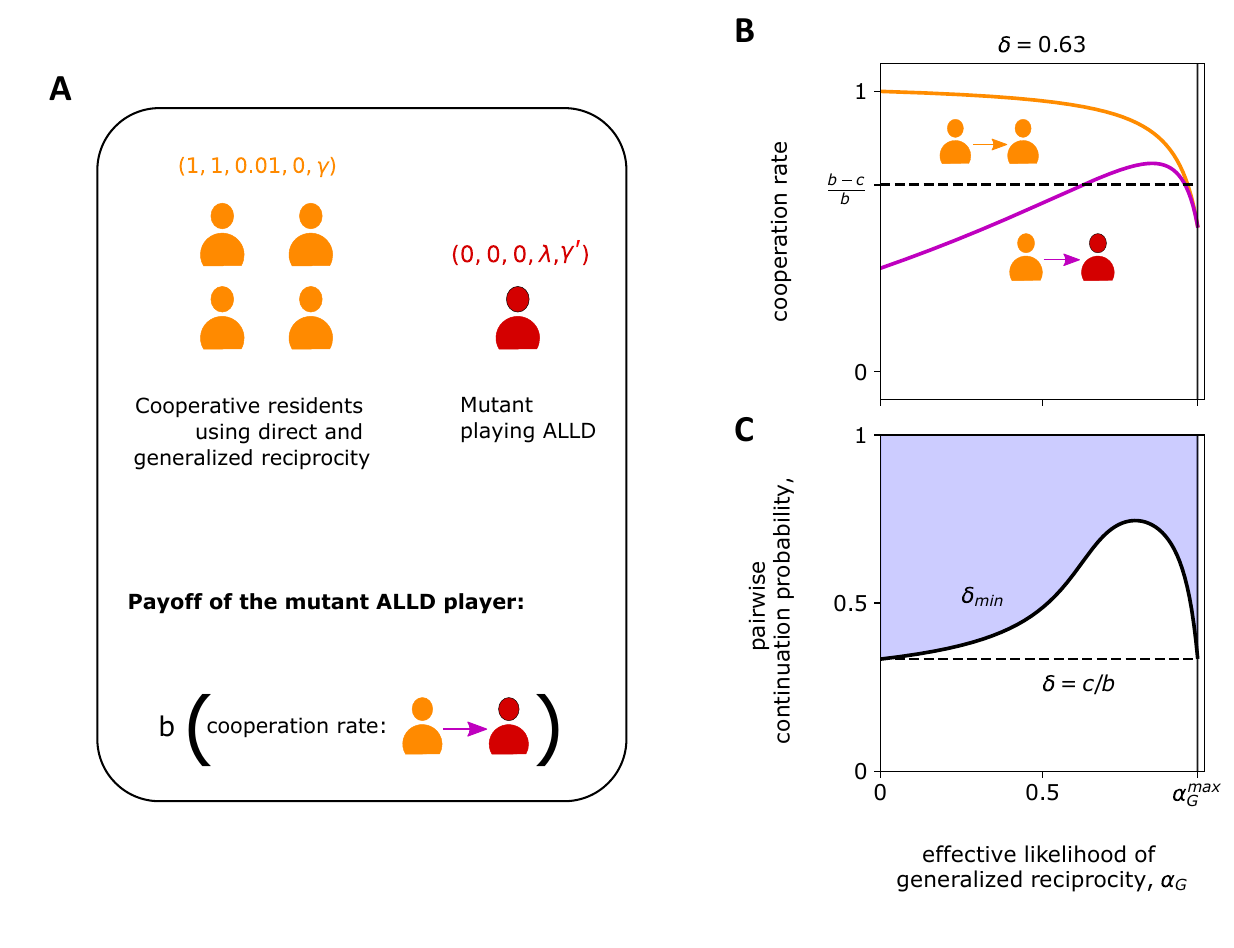}
    \caption{\textbf{A mixture of direct and generalized reciprocity may be less conducive to cooperation than either pure direct or pure generalized reciprocity (A),} We consider the situation where an always defecting  (\textrm{ALLD}) mutant is introduced into a homogenous population of cooperative residents who play the strategy $(1,1,q,0,\gamma)$. Before the mutant is introduced, the residents receive an expected payoff of $b-c$. To be a Nash equilibrium, the average probability that residents cooperate with the mutant must not exceed $(b-c)/b$. \textbf{(B),} The way the residents react to each other (orange curve) and to the \textrm{ALLD} mutant (pink curve) depends on how often they use generalized reciprocity to update reputations, $\alpha_\mathrm{G}$. When the residents use pure direct reciprocity, i.e., $\alpha_\mathrm{G} = 0$, they employ generous tit-for-tat and discriminate against the \textrm{ALLD} mutant. They cooperate fully with other residents but rarely forgive the mutant. As $\alpha_\mathrm{G}$ increases, the two curves tend to converge. When residents use pure generalized reciprocity, i.e., $\alpha_\mathrm{G} = \alpha_\mathrm{G}^{max}$, they treat everyone equally. If the pairwise continuation probability, $\delta$, is not sufficiently high (for example in this case where $\delta = 0.63$), the \textrm{ALLD} mutant earns a payoff greater than $b-c$ for certain $0< \alpha_\mathrm{G} < 1$. \textbf{(C),} Depending on $\alpha_\mathrm{G}$, the minimum value of the pairwise continuation probability, $\delta_{min}$, that is necessary for a resident strategy to be a best response to itself, varies. It is at its lowest, $\delta_{min} = c/b$, when residents use pure direct or pure generalized reciprocity. Parameters: $b = 3$, $c = 1$ and $n=50$.}
    \label{fig:GR-equilibrium}
\end{figure}

\noindent As a result, stable cooperation via a mixture of direct and generalized reciprocity requires more interactions than stable cooperation by either pure direct or pure generalized reciprocity (Supplementary Fig~\ref{fig:GR-equilibrium}C). 

\section{Effective likelihood of direct, indirect and generalized reciprocity}

In this section we introduce three probabilities to describe the effective likelihood with which players use direct, indirect and generalized reciprocity in a homogenous population. We define these probabilities below, 

\begin{enumerate}
    \item \textbf{The effective likelihood of direct reciprocity}, $\alpha_\mathrm{D}$: the probability that the reputation of an arbitrary individual $j$ from the perspective of another arbitrary individual $i$ is not altered between two direct interactions involving $i$ and $j$. 
    
    \item \textbf{The effective likelihood of indirect reciprocity}, $\alpha_\mathrm{I}$: the probability that between two direct interactions involving $i$ and $j$, $i$ updates $j$'s reputation \emph{and}, the last update he makes about $j$'s reputation before meeting her is due to $j$'s interaction with a third party $k \neq i$. 
    
    \item \textbf{The effective likelihood of generalized reciprocity}, $\alpha_\mathrm{G}$:  the probability that between two direct interactions involving $i$ and $j$, $i$ updates $j$'s reputation \emph{and}, the last update he makes about $j$'s reputation before meeting her is due to $i$'s interaction with a third party $k \neq i$. 
\end{enumerate}

\noindent By definition then $\alpha_\mathrm{D} + \alpha_\mathrm{I} + \alpha_\mathrm{G} = 1$. This normalizing condition allows us to quantify how much weight players in a homogenous population put on direct, indirect and generalized reciprocity. In a population where all individuals play $(y,p,q,\lambda,\gamma)$, the effective likelihood of direct reciprocity, $\alpha_\mathrm{D}$, can be calculated as

\begin{equation}
\begin{split}
    \alpha_\mathrm{D} &= \frac{1}{\binom{n}{2}} \cdot \sum_{s = 0}^{\infty} \sum_{l + m = 0}^{s} \frac{s\,!}{l\,! m\,! (s-l-m)\,!} \, \cdot \zeta_1^{l+m} \, \cdot \zeta_2^{s-l-m} \cdot (1 - \lambda)^l (1 - \gamma)^m \\ \\
    &= \frac{1}{1+(n-2)(\lambda + \gamma)}.
\end{split}
\label{Eq:direct_reciprocity_TL}
\end{equation}
\\
\noindent Where $\zeta_1$ is the probability that either $i$ or $j$ (not both) is part of an interaction at a particular round and $\zeta_2$ is probability that neither $i$ nor $j$ is part of a particular interaction. These probabilities are given by, 
\begin{equation}
\begin{split}
    \zeta_1 &=  \frac{n-2}{\binom{n}{2}},\\ \\
    \zeta_2 &= \frac{\binom{n-2}{2}}{\binom{n}{2}}.
\end{split}
\label{s_1_s_2}
\end{equation}

\noindent The probability that between two direct interactions involving $i$ and $j$, the reputation of $j$ with respect to $i$ is not updated depends on the likelihood of two events. Either neither of them are involved in an interaction between these two interactions (in this way there is no opportunity to update reputations) or, one of them is part of an interaction, but the reputation consequences for $j$ because of this interaction are ignored. In the expression inside the sums in Eq.~(\ref{Eq:direct_reciprocity_TL}), we consider that out of the $s$ interactions between two direct interactions involving $i$ and $j$, $l$ interactions involve only $j$, $m$ interactions involve only $i$, and, $s-l-m$ interactions involve neither $i$ nor $j$. These events jointly occur with probability from the multinomial distribution, 
$$\frac{s\,!}{l\,! m\,! (s-l-m)\,!} \cdot \zeta_1^{l+m} \zeta_2^{s-l-m}$$. \\
\noindent In all the $l$ interactions in which only $j$ is involved, $i$ ignores $j$'s action with probability $1-\lambda$. Similarly in all the $m$ interactions involving only $i$, $i$ does not update $j$'s reputation based on what is done to him, with probability $1-\gamma$. Thus, the probability that $i$ makes no update to $j$'s reputation in this exact distribution of events (i.e., for particular values of $s$, $l$ and $m$) is, 

$$\frac{s\,!}{l\,! m\,! (s-l-m)\,!} \, \cdot \zeta_1^{l+m} \, \cdot \zeta_2^{s-l-m} \cdot (1 - \lambda)^l (1 - \gamma)^m.$$
\\
\noindent To compute $\alpha_\mathrm{D}$ we sum up the above probability for all values of $l$ and $m$ that sum up to $s$ and then sum for all values of $s$. The probability $1/(\binom{n} {2})$ on the outside of the sum in Eq.~(\ref{Eq:direct_reciprocity_TL}) is the probability that $i$ and $j$ are chosen for their second direct interaction. \\

\noindent The probability $1 - \alpha_\mathrm{D}$ is the probability that the reputation of $j$ from the perspective of $i$ is updated between two direct interactions involving $i$ and $j$. The final reputation of $j$ from $i$'s perspective before they meet could be because of an interaction that involved $j$ but not $i$ (indirect reciprocity) or an interaction that involved $i$ but not $j$ (generalized reciprocity). Then, the likelihood $\alpha_\mathrm{I}$ and $\alpha_\mathrm{G}$ can be expressed as 

\begin{equation}
\begin{split}
    \alpha_\mathrm{I} &= (1 - \alpha_\mathrm{D}) \cdot \zeta_3 \cdot \lambda, \\ 
    \alpha_\mathrm{G} &= (1 - \alpha_\mathrm{D}) \cdot \zeta_3 \cdot \gamma.
\end{split}
\end{equation}
\\
\noindent where $\zeta_3$ is a constant of proportionality. Now as, $\alpha_\mathrm{I} + \alpha_\mathrm{G} = 1 - \alpha_\mathrm{D}$, 
\begin{equation}
    (1 - \alpha_\mathrm{D}) = (1 - \alpha_\mathrm{D}) \cdot \zeta_3 \cdot (\lambda + \gamma) \\[10pt]
\end{equation}
\noindent Since $1 - \alpha_\mathrm{D}$ is non-zero only when $\lambda, \gamma > 0$, 

\begin{equation}
\alpha_\mathrm{I} = \begin{cases}
(1 - \alpha_\mathrm{D}) \cdot \frac{\lambda}{\lambda + \gamma} \quad &\text{ if }\lambda + \gamma \neq 0 \\[10pt]
0 \quad &\text{otherwise }
\end{cases}
\end{equation}
\noindent and, 
\begin{equation}
\alpha_\mathrm{G} = \begin{cases}
(1 - \alpha_\mathrm{D}) \cdot \frac{\gamma}{\lambda + \gamma} \quad &\text{ if }\lambda + \gamma \neq 0 \\[10pt]
0 \quad &\text{otherwise. }
\end{cases}
\end{equation}
\\[10pt]
\noindent \textbf{Theoretical minimum of the likelihood of direct reciprocity:} In this model, as long as $\delta > 0$, there is a non-zero probability that direct reciprocity is employed in updating reputations. This is because there is always a non-zero chance that a pair of players will meet again after they meet in the present \emph{and} in-between their two direct interactions, one of them will not update the other's reputation (either because neither of them participate in an interaction or the player choose to ignore the other's interactions). From Eq.~(\ref{Eq:direct_reciprocity_TL}), we can see that this theoretical minimum occurs when $\lambda + \gamma = 1$,
\begin{equation}
    \alpha_\mathrm{D}^\mathrm{min}  = \frac{1}{n-1}. \\[10pt]
    \label{Eq:DR-theoretical-min}
\end{equation}
\noindent This also imposes a theoretical maximum on the likelihood of indirect reciprocity and the likelihood of generalized reciprocity which occur at $(\lambda,\gamma)$ = $(1,0)$ and $(0,1)$ respectively, 
\begin{equation}
      \alpha_\mathrm{I}^\mathrm{max}  = \alpha_\mathrm{G}^\mathrm{max} = 1 - \alpha_\mathrm{D}^\mathrm{min} = \frac{n-2}{n-1}.
\end{equation}
\noindent Any tuple of the effective likelihoods, $(\alpha_\mathrm{D}, \alpha_\mathrm{I}, \alpha_\mathrm{G})$ can be expressed as a convex combination of the corners of the triangular simplex, $\mathcal{S}_{\mathrm{DIG}}(n)$, formed by the points: $(1,0,0)$, $(\alpha_\mathrm{D}^\mathrm{min}, \alpha_\mathrm{I}^\mathrm{max},0)$, $(\alpha_\mathrm{D}^\mathrm{min}, 0, \alpha_\mathrm{G}^\mathrm{max})$. For a fixed population size, $n$, every tuple $(\lambda, \gamma) \in [0,1]^2$ maps to a unique point inside the simplex, $\mathcal{S}_{\mathrm{DIG}}(n)$.  Formally, every triplet $(n, \lambda, \gamma) \in \mathbb{Z}_{\geq 2} \times [0,1]^2$ has a unique map in  $\mathcal{S}_{\mathrm{DIG}}(n)$. That is, there is a one-to-one mapping between the two spaces. 
\section{Numerically efficient computation of payoffs}
For evolutionary simulations, we numerically compute the players' payoffs. The payoffs in a homogenous population can be explicitly computed using the result in Proposition ~\ref{Prop:coop-payoff-homogenous}. However, in heterogenous populations, one needs to calculate the payoffs of players using Eq.~(\ref{Eq:payoff-equation}). This requires computing the inverse of a matrix of size $n(n-1)$, which could result in large computation times for large populations. In this section we discuss a more efficient way to compute payoffs. This efficient way takes advantage of the fact that players using the same strategy have the same reputation from the perspective of all other players. 
\\

\noindent We consider that there are $s$ distinct strategies in the population (indexed from $1$ to $s$) and the sizes of these sub-populations are $k_1, k_2,..,k_s$. The $i$-th strategy plays $\sigma_i := (y_i, p_i, q_i, \lambda_i, \gamma_i)$. The recursion equation describing the change in the probability that players consider each other \textit{good} Eq.~(\ref{Eq: recursion_2}), can now be written to describe the change in the probability that \emph{strategies} consider each other \textit{good}. In the following recursions, $x_{ij}(t)$ indicates the probability that at time $t$, a player using strategy $i$ considers a player using strategy $j$ as \textit{good}, 

\begin{equation}
\begin{split}
x_{ii}(t + 1) =  &(1 - w - (\bar{w} - w)(\lambda_i + \gamma_i) + wr_i + (k_i - 2)w r_i \lambda_i (1 - 2 \varepsilon) + (k_i - 2)w r_i \gamma_i)x_{ii}(t)  \\[5pt]
&+ w\lambda_i r_i (1 - 2 \varepsilon) \sum_{l \neq i}^s k_l x_{il}(t) \\[5pt]
&+ w \gamma_i r_i \sum_{l \neq i}^s k_l x_{li}(t) \\[5pt] 
&+ wq_i + w \gamma_i q_i (n-2) + w \lambda_i (q_i + \varepsilon r_i) (n-2) \\[20pt]
x_{ij}(t + 1) =  &(1 - w - (\bar{w} - w)(\lambda_i + \gamma_i))x_{ij}(t)  \\[5pt]
&+ (wr_i + w \lambda_i r_i (1 - 2 \varepsilon)(k_i - 1) + w \gamma_i r_i (k_j - 1))x_{ji}\\
&+ w\lambda_i r_i (1 - 2 \varepsilon) \sum_{l \neq i}^s (k_l - 1_{jl}) x_{jl}(t) \\
&+ w \gamma_i r_i \sum_{l \neq i,j}^s k_l x_{li}(t) \\ 
&+ w \gamma_i r_i (k_i - 1)x_{ii} \\[5pt] 
&+ wq_i + w \gamma_i q_i (n-2) + w \lambda_i (q_i + \varepsilon r_i) (n-2)
\end{split}
\label{Eq:efficient-recursion}
\end{equation}
\\
\noindent Here $r_i := p_i - q_i$ and $1_{ij}$ is the Kronecker delta function which is $1$ if $i=j$ and 0 if $i \neq j$. In our evolutionary simulations we always compute payoffs for populations with at most two strategies (that is, in the limit of rare mutations in the Imhof-Nowak process~\cite{Imhof:PRSB:2010, fudenberg:JET:2006}). We describe the remaining steps of the efficient computation method for only this special case. However, one can already see that recursions have a general form for any $s$ in Eq.~(\ref{Eq:efficient-recursion}). Using the same technique, it is also possible to construct the remaining steps for arbitrary number of strategies in the population. \\

\noindent We assume that the two strategies in the population are $\sigma_1 = (y_1, p_1, q_1, \lambda_1, \gamma_1)$ and $\sigma_2 = (y_2, p_2, q_2, \lambda_2, \gamma_2)$. The expected probability that strategies consider each other \textit{good} at a randomly selected round can be computed as,  

\begin{equation}
    \mathbf{x} = (\mathbb{1} - d\mathbf{M})^{-1}((1 - d)\mathbf{x_0} + d\mathbf{v}) \\[10pt]
\label{easy_reputation_inverse}
\end{equation}

\noindent Here $\mathbf{x}$ is a column vector of size 4,

\begin{equation}
\mathbf{x} = \begin{pmatrix}
x_{11} \\
x_{12} \\
x_{21}\\
x_{22}\\
\end{pmatrix}
\\[10pt]
\end{equation}

\noindent and the $4 \times 4$ square matrix $\mathbf{M} = (m_{ab,cd})_{a,b,c,d \in \{1,2\}}$, the four-dimensional vectors $\mathbf{x}_0$ and $\mathbf{v}$ are given by

\begin{equation}
m_{ab,cd} = \begin{cases}
1-w+wr_a-(\bar{w}-w)(\lambda_a + \gamma_a) \\ + wr_a(k_a - 2)\lambda_a(1-2\varepsilon)+ wr_a(k_a-2)\gamma_a &\quad \text{ if } a=b=c=d,\\
1-w-(\bar{w}-w)(\lambda_a + \gamma_a) &\quad \text{ if } a=c\neq b = d,\\ 
0 &\quad \text{ if }  a=b \neq c=d,\\
w\lambda_ar_a(1-2\varepsilon)k_d &\quad \text{ if }  a=b=c \neq d,\\
w\gamma_ar_ak_c &\quad \text{ if } a=b=d \neq c,\\
w\gamma_ar_a(k_a-1) &\quad \text{ if } a=c=d \neq b,\\
wr_a + wr_a\gamma_a(k_b-1) \\+ wr_a\lambda_a(1-2\varepsilon)(k_a-1) &\quad \text{ if } a=d \neq b = c,\\
w\lambda_ar_a(1-2\varepsilon)(k_b-1) &\quad \text{ if } a \neq b = c = d.
\end{cases}
\\[20pt]
\end{equation}

\begin{equation}
\mathbf{x_0} = \begin{pmatrix}
y_1\\
y_1\\
y_2\\
y_2\\
\end{pmatrix}, \\
\end{equation}
\\
\begin{equation}
\mathbf{v} = \begin{pmatrix}
wq_1 + w \gamma_1 q_1 (n-2) + w \lambda_1 (q_1 + \varepsilon r_1)(n-2)\\ \\
wq_1 + w \gamma_1 q_1 (n-2) + w \lambda_1 (q_1 + \varepsilon r_1)(n-2)\\ \\
wq_2 + w \gamma_2 q_2 (n-2) + w \lambda_2 (q_2 + \varepsilon r_2)(n-2)\\ \\
wq_2 + w \gamma_2 q_2 (n-2) + w \lambda_2 (q_2 + \varepsilon r_2)(n-2)\\ \\
\end{pmatrix}
\end{equation}

\noindent respectively. Payoff for players using strategy $i$ can then be calculated from the computed $\mathbf{x}$ from Eq.~(\ref{easy_reputation_inverse}) using the expression,
\begin{equation}
    \pi_i = \sum_{j = 1}^{s} \frac{k_j - 1_{ij}}{n-1} (x_{ji}b - x_{ij}c).
\label{efficient_payoff}
\end{equation}
\\
\noindent For $s$ strategies in the population, the method of efficiently computing payoffs requires the computation of the inverse of a matrix of size $s^2$ (as opposed to the original method where the matrix size was $n(n-1)$).
\clearpage
\section*{Supplementary Figures}
\begin{figure}[h!]
    \centering
    \includegraphics[width = \textwidth]{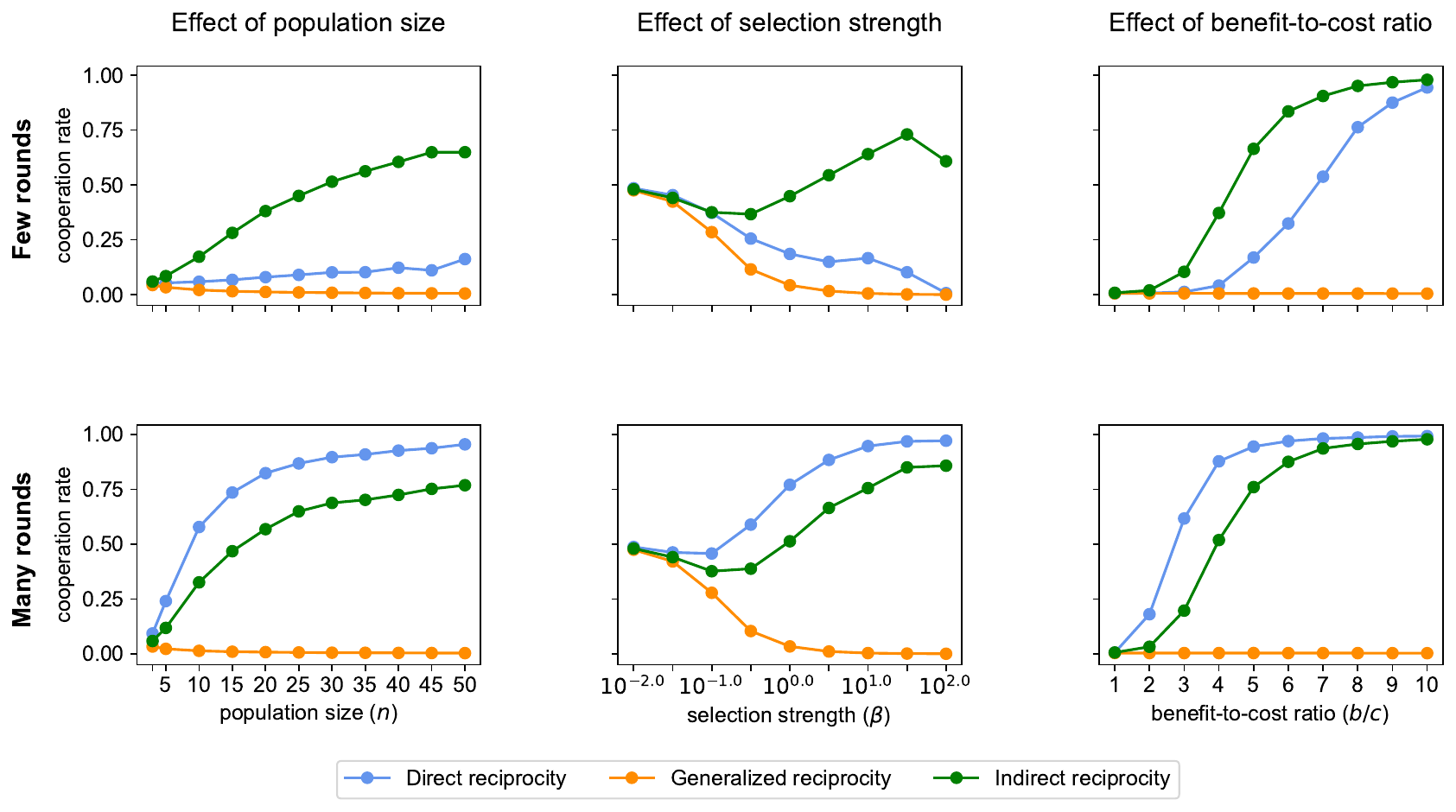}
    \vspace{0.2cm}
    \caption{\textbf{The evolution of cooperation when populations rely on a single mode of reciprocity and rarely explore.} We explore how the effect of Figure 3 (from main text) is influenced by changing different parameters of the model. Specifically, we examine the effects of population size, selection strength, and the benefit-to-cost ratio. We consider two scenarios: one with a few rounds ($\delta = 0.5$, top row) and one with many rounds  ($\delta = 1$, bottom row). Each data point represents the average cooperation rate of the simulation process. In the bottom panel, with many rounds, it is the average of a single simulation. For the top panel, it is the average of three independent simulations. Each simulation stops after $10^7$ mutant strategies have been introduced into the population. If a parameter is not varying, it is fixed to the same value we use for Figure 3. That is, population size
    $n = 50$ and benefit of cooperation, $b = 5$, cost of cooperation, $c=1$, and selection strength, $\beta = 10$. For all these simulations, we assume that there is no error in observing third party interactions, i.e, $\varepsilon=0$. We observe that cooperation fails to evolve via generalized reciprocity across a range of values of $n, \beta$ and $b/c$.}
    \label{fig:extended-figure-single-reciprocity}
\end{figure}

\begin{figure}[h!]
    \centering
    \includegraphics[width = \textwidth]{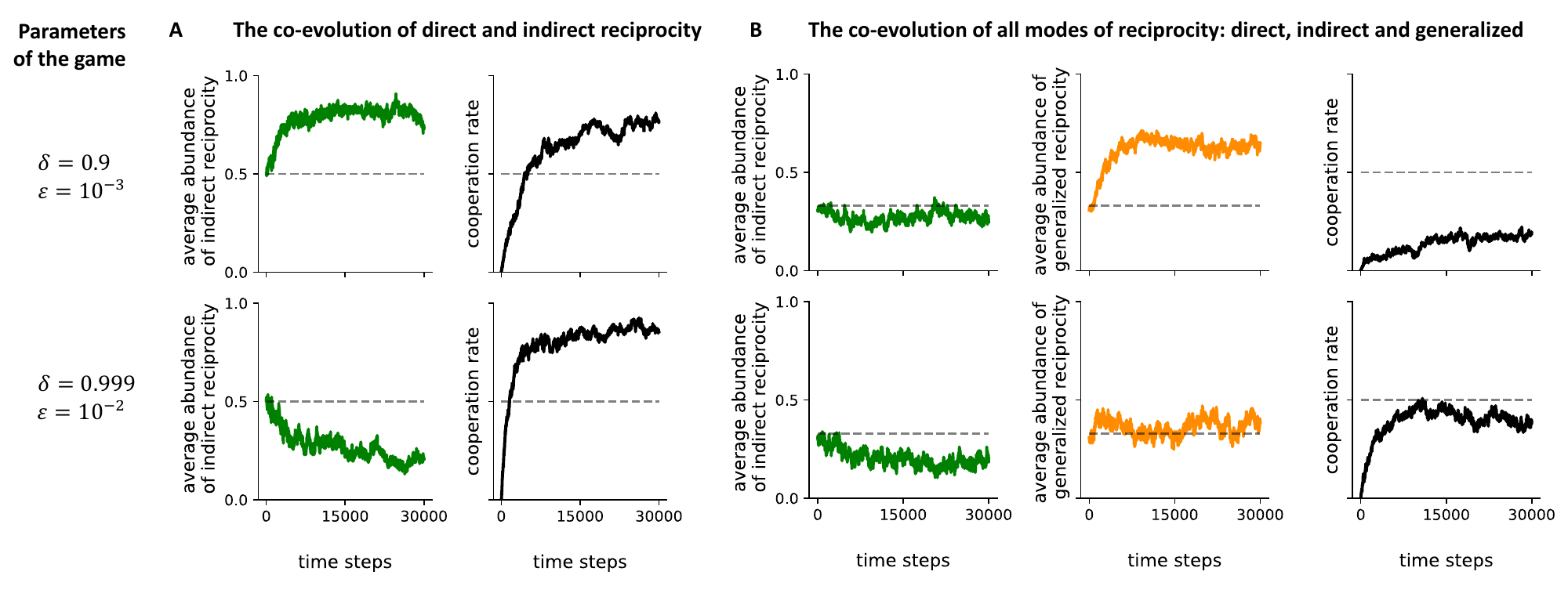}
    \caption{\textbf{Evolutionary dynamics with and without generalized reciprocity demonstrates the impact of generalized reciprocity on cooperation.} Here we present average evolutionary trajectories for two kinds of evolutionary dynamics as in Fig.4 from main text. We show the time series of the likelihood that players use indirect reciprocity ($\alpha_\mathrm{I}$), the likelihood that they use generalized reciprocity ($\alpha_\mathrm{G}$) in a population, averaged over 150 independent simulations. In simulations performed for \textbf{(A)}, individuals can only adopt direct and indirect reciprocity. However in \textbf{(B)}, they can adopt all the three modes of reciprocity. We also plot the average cooperation in populations as they evolve over time. We show these time series for two parameter combinations: $\delta = 0.9, \varepsilon = 10^{-3}$ (top row) and $\delta = 0.999, \varepsilon = 10^{-2}$ (bottom row). Cooperation rates remain considerably low in the evolutionary process where individuals can adopt generalized reciprocity.}
    \label{fig:treatment-comparison-timeseries-evolution}
\end{figure}

\begin{figure}[t!]
    \centering
    \includegraphics[width = \textwidth]{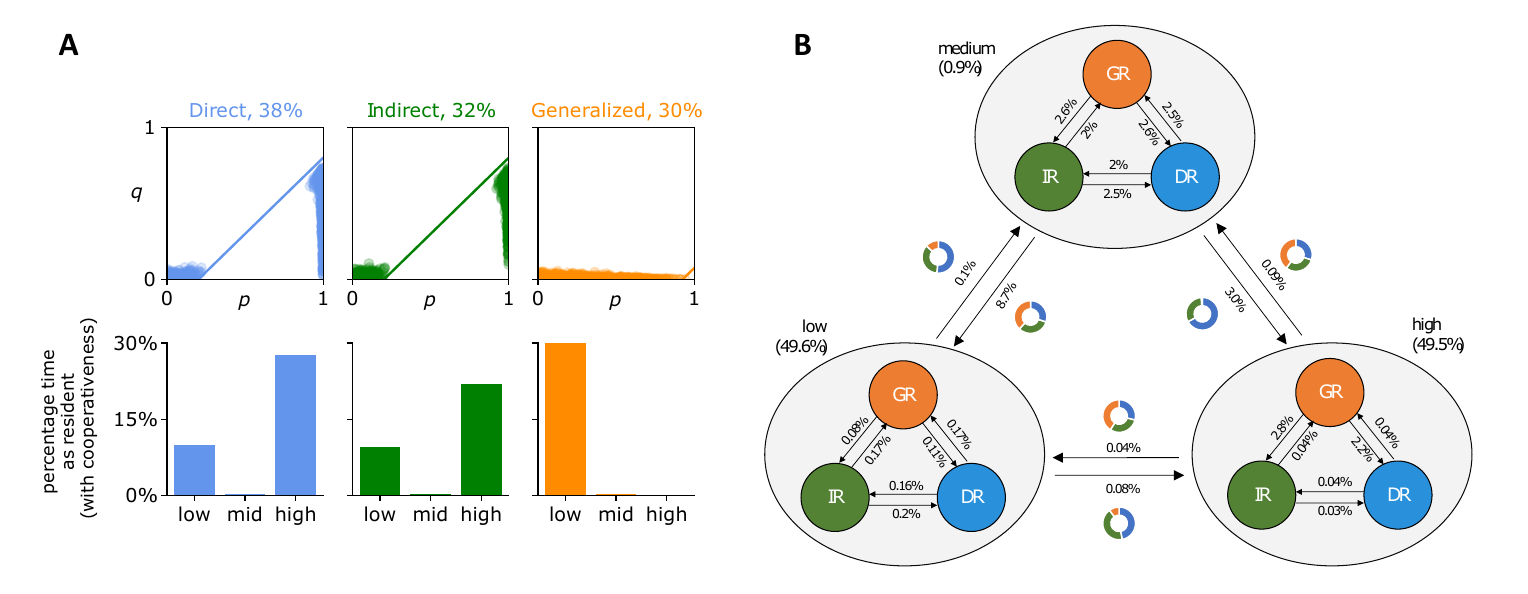}
    \caption{\textbf{The analysis of a simulation with many pairwise interactions and rare observational errors reveals a mechanism behind the breakdown of cooperation in the presence of generalized reciprocity.} We conduct an evolutionary simulation at low mutation rate with the parameters $\delta = 0.999$, and $\varepsilon = 10^{-3}$. Individuals in the simulation were free to adopt any of the three modes of reciprocity. \textbf{(A),} Among the 47,000 resident populations that emerged, 38\% used direct reciprocity, 32\% used indirect reciprocity and the rest, 30\%, used generalized reciprocity. We show the strategy elements $p$ and $q$ of the 1000 most successful residents from each reciprocity class using colored dots. We categorize each resident population into one of the three classes: low (less than 1/3), medium (between 1/3 and 2/3) and high (more than 2/3), based on the probability of observing cooperation in them. The most abundant residents are high DR and low GR. \textbf{(B),} In this illustration we show the estimated transition probabilities between the cooperative classes (arrows between the 3 large ovals) and estimated transition probabilities between residents that use different modes of reciprocity but fall under the same cooperative class (arrows between the colored circles inside each oval). We estimate these transition probabilities using data from our simulation. For each transition between the cooperative classes, we denote the relative contribution of each mode of reciprocity for that transition using colors on a circular scale (blue for direct reciprocity, green for indirect reciprocity and orange for generalized reciprocity). All other parameters are the same as in Fig. 3 of main text. 
    }
    \label{fig:GR-reason-breakdown}
\end{figure}


\newpage
\begin{figure}[h!]
    \centering
    \includegraphics[width = 0.5\textwidth]{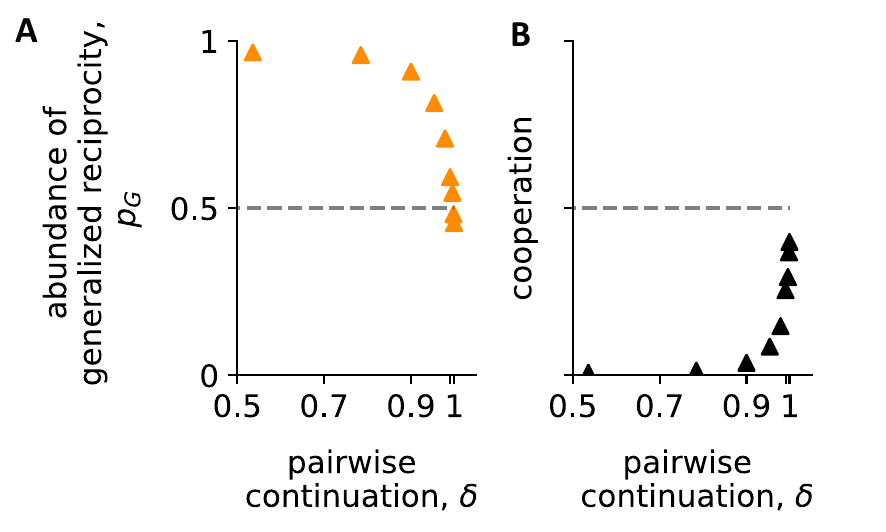}
    \vspace{-0.2cm}\caption{\textbf{The evolutionary dynamics of direct and generalized reciprocity.} In the evolutionary simulations that we conduct for this Figure, individuals can only choose to adopt either pure direct reciprocity or pure generalized reciprocity. The evolutionary process is conducted for rare mutations. In all simulations, the evolutionary process begins with a strategy that always defects. In \textbf{(A)}, we show the average time series of the likelihood that individuals use indirect reciprocity, $\alpha_\mathrm{I}$ (which is trivially always zero since we restrict our simulations to exclude indirect reciprocity) and the likelihood that they use generalized reciprocity, $\alpha_\mathrm{G}$ to update reputations in homogenous resident populations. We also plot the evolution of cooperation over time in the corresponding resident populations. We show these time series for two values of pairwise continuation probability in the population, $\delta = 0.9$ and $\delta = 0.999$. In \textbf{(B)}, We plot the time average of $\alpha_\mathrm{I}$, $\alpha_\mathrm{G}$ and cooperation rate in populations over long simulations for values of $\delta$ ranging between 0.5 and 0.999. The likelihood that individuals use generalized reciprocity drops with more expected rounds. The cooperation rate grows simultaneously.} 
    \label{fig:evolution-with-and-without-IR}
\end{figure}
\clearpage

\begin{figure}
    \centering
    \includegraphics[width = \textwidth]{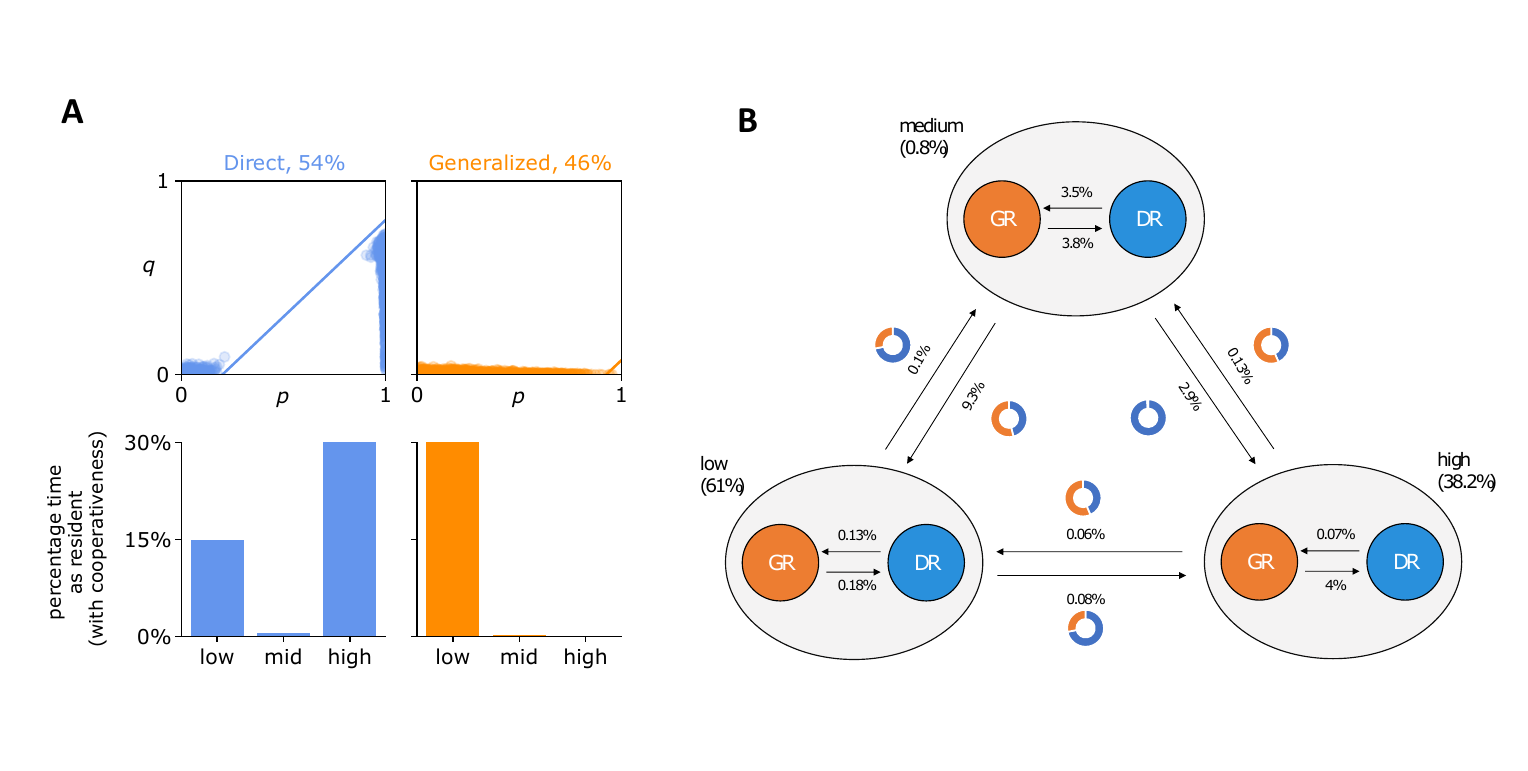}
    \caption{\textbf{The distribution of strategies and transitions between them in an evolutionary simulation where individuals were only allowed to select between pure direct and pure generalized reciprocity.} For this figure, evolutionary simulations were performed for a high value of pairwise continuation probability, $\delta = 0.999$. The methods behind the simulation and the related parameter values are the same as the simulation used to create Fig. 5 in main text.} 
    \label{fig:IR-off-high-continuation}
\end{figure}

\clearpage
\section*{Appendix: Proofs}

\begin{proof}[Proof of Proposition ~\ref{Prop:coop-payoff-homogenous}]
\noindent We define a function $f_{ij}(T)$ in the following way,

\begin{equation}
    f_{ij}(T) := (1-d) \sum_{\tau = 0}^T d^\tau (dx_{ij}(\tau + 1) - x_{ij}(\tau)) 
\end{equation}\\
In the above definition we first insert the recursive expression of $x_{ij}(\tau + 1)$ in terms of $x_{ij}(\tau)$ from Eq.~(\ref{Eq: recursion_2}) and then evaluate $\lim_{T \to \infty} f_{ij}(T)$. While doing the steps, we use the definition average reputation of $j$ with respect to $i$, $x_{ij}$, from Eq.~(\ref{Eq:discounted_x}) to simplify and get,\\\\
\begin{equation}
\begin{split}
    \lim_{T \to \infty} f_{ij}(T) &= (d-dw-d(\bar{w}-w)(\lambda_i + \gamma_i))x_{ij} + dwr_ix_{ji} + dw\lambda_ir_i(1-2\varepsilon)\displaystyle\sum_{l \neq i,j}x_{jl} \\[10pt] &+ dw\gamma_i r_i \displaystyle \sum_{l \neq i,j} x_{li} + dwq_i + dw\gamma_iq_i(n-2) + dw\lambda_i(q_i + \varepsilon r_i)(n-2) \\[10pt] &- x_{ij} 
\end{split}
\label{Eq:limit-first-expression}
\end{equation}
There is way to evaluate another expression for $\lim_{T \to \infty} f_{ij}(T)$. This is by observing that the sum in the definition of $f_{ij}(T)$ is a telescopic sum,
\begin{equation}
\begin{split}
f_{ij}(T):&= (1-d) \sum_{\tau = 0}^T d^\tau (dx_{ij}(\tau + 1) - x_{ij}(\tau)) \\[10pt]
&=(1-d) \displaystyle (d^{T + 1}x_{ij}(T+1) - x_{ij}(0))
\end{split}    
\end{equation}
and therefore,
\begin{align}
\lim_{T \to \infty} f_{ij}(T)T:&= -(1-d)x_{ij}(0) = -(1-d)y_i 
\label{Eq:limit-second-expression}
\end{align}
Comparing Eq.(\ref{Eq:limit-first-expression}) and Eq.(\ref{Eq:limit-second-expression}) we equate the two right hand sides to get,
\begin{equation}
\begin{split}
 &(d - dw - d(\bar{w}-w)(\lambda_i + \gamma_i)-1)x_{ij} + dwr_ix_{ji} + dw\lambda_ir_i(1-2\varepsilon)\displaystyle\sum_{l \neq i,j}x_{jl} \\[10pt] &+ dw\gamma_i r_i \displaystyle \sum_{l \neq i,j} x_{li} + dwq_i + dw\gamma_iq_i(n-2) + dw\lambda_i(q_i + \varepsilon r_i)(n-2)  = -(1-d)y_i\\[10pt] 
\end{split}
\label{Eq:master-equation-work}
\end{equation}
In a homogenous population playing the strategy $\sigma = (y,p,q,\lambda,\gamma)$, $x_{ij} = x_{ji} = x_{jl} = x_{li} = x$. Then, Eq.~(\ref{Eq:master-equation-work}) reduces to, 

\begin{equation}
\begin{split}
  &(d - dw - d(\bar{w}-w)(\lambda + \gamma)-1)x + dwrx + dw\lambda r (1-2\varepsilon)(n-2)x \\[10pt] &+ dw\gamma r (n-2)x + dwq + dw\gamma q(n-2) + dw\lambda(q + \varepsilon r)(n-2)  = -(1-d)y\\[20pt]   
  \implies x &= \frac{(1-d)y + dwq + d(\bar{w} - w)(q\lambda + q\gamma + \lambda \varepsilon r)}{(1 - d) + dw(1-r) + d(\bar{w} - w)((\lambda + \gamma)(1 - r) + 2 \varepsilon \lambda r)} \\ \\
\end{split}
\end{equation}
Now using the expressions $w = 2/(n(n-1))$ and $\bar{w} = 2/n$ and the expression for $\delta$ in Eq.(\ref{Eq:delta-with-d-n}), one can simplify $x$ as, \\[5pt]
\begin{equation}
    x = \frac{(1 - \delta)y + \delta q + \delta (n-2)(q(\lambda + \gamma) + \lambda \varepsilon r)}{1 - \delta r + \delta (n-2)((\lambda + \gamma)(1 - r) + 2 \lambda \varepsilon r)} \\[10pt]
\label{Eq:proof-x-homogenous}
\end{equation}
Here $r = p-q$. The payoff of every player in this homogenous population is $(b-c)x$. For $0 < \delta < 1$, $n>2$ and $\varepsilon > 0$, one can observe from the numerator of the expression in Eq.~(\ref{Eq:proof-x-homogenous}) that $x = 0$ (a fully defective population) if and only if either a) $y=q=0$ and $p=0$ or b) $y=q=0$ and $\lambda = 0$. Here $y=q=0$ is a necessary condition for $x = 0$, and, together with either $p=0$ or $\lambda = 0$ it is also sufficient. For $0 < \delta < 1$, $n>2$ and $\varepsilon = 0$, $y=q=0$ is necessary as well as sufficient for $x=0$. The same is true for $n = 2$. \\

\noindent Similarly, now if we plug $x = 1$ (a fully cooperative population) in Eq.~(\ref{Eq:proof-x-homogenous}), we get,
\begin{equation}
    (1-\delta)y = 1 - \delta p + \delta (n-2)((\lambda+\gamma)(1-p) + \lambda \varepsilon (p-q))
\end{equation}
When $n = 2$, the only solution for the above equation is $y = p = 1$. However, when $n > 2$, the left hand side in itself is always less than or equal to $1-\delta$ while the right hand side in itself is always greater than or equal to $1-\delta$. When $0 < \delta < 1$, the only case when they are exactly equal to $1-\delta$ is when $y = p = 1$ and either $q=1$ or $\lambda = 0$. 
\end{proof}

\begin{proof}[Proof of Lemma ~\ref{claim: first linear relationship}]
In the case where there is one mutant in a population of size labeled as player $n$, we denote the average probability that a resident considers another resident \textit{good} with $x_{ik}$, and that a resident considers the mutant \textit{good} as $x_{in}$ and that the mutant considers the resident \textit{good} as $x_{ni}$. The resident plays the strategy $(y,p,q,\lambda,\gamma)$. Then, rewriting Eq.~(\ref{Eq:master-equation-work}) from the perspective of a resident $i$, towards the mutant $n$, i.e., by replacing $j$ with $n$ in Eq.~(\ref{Eq:master-equation-work}) we get, 
\begin{equation}
\begin{split}
 &(d - dw - d(\bar{w}-w)(\lambda + \gamma)-1)x_{in} + dwrx_{ni} + dw\lambda r(1-2\varepsilon)(n-2)x_{ni} \\[10pt] &+ dw\gamma r (n-2) x_{ik} + dwq + dw\gamma q(n-2) + dw\lambda(q + \varepsilon r)(n-2)  = -(1-d)y\\[10pt]
\end{split}
\end{equation}
Rearranging the terms of the above equation so that the coefficient of $x_{in}$ is 1, we get, 
\begin{equation}
    x_{in} = B_1 + B_2x_{ni} + B_3x_{ik}
    \label{eq:first-linear}
\end{equation}
where $B_{1}, B_2$ and $B_3$ are given in Eq.~(\ref{Eq:coefficients-Bs}). If we rewrite again Eq.~(\ref{Eq:master-equation-work}) but this time from the perspective of a resident $i$, towards another resident $k$, i.e., replacing $j$ with $k$ in Eq.~(\ref{Eq:master-equation-work}) we get, 
\begin{equation}
\begin{split}
    &(d - dw - d(\bar{w}-w)(\lambda + \gamma)-1)x_{ik} + dwrx_{ik} + dw\lambda r(1-2\varepsilon)(n-3)x_{ik} \\[10pt] &+ dw\gamma r (n-3) x_{ik} + dw\lambda r (1-2\varepsilon)x_{in} + dw\gamma r x_{ni} + dwq + dw\gamma q(n-2) + \\[10pt] &+dw\lambda(q + \varepsilon r)(n-2)  = -(1-d)y\\[10pt] 
\end{split}
\end{equation}
Rearranging the terms of the above equation so that the coefficient of $x_{ik}$ is 1, we get, 
\begin{equation}
    x_{ik} = C_1 + C_2x_{ni} + C_3x_{in}
    \label{eq:second-linear}
\end{equation}
where $C_{1}, C_2$ and $C_3$ are given in Eq.~(\ref{Eq:coefficients-Bs}).
\end{proof}
\newpage
\begin{proof}[Proof of Proposition~\ref{Prop:average-payoff-mutant-linear} ]
The linear equations in the previous lemma hold only if $n>2$. In this case, combining Eq.~(\ref{eq:first-linear}) and Eq.~(\ref{eq:second-linear}) and eliminating $x_{ik}$, one gets for $n>2$,

\begin{equation}
x_{in} = \left( \frac{B_1 + B_3C_1}{1-B_3C_3} \right) + \left( \frac{B_2 + B_3C_2}{1-B_3C_3} \right) x_{ni} \\[10pt]
\end{equation}
\noindent For the simpler case of $n=2$, the relationship between $x_{in}$ and $x_{ni}$ is only due to direct reciprocity and is given by,
\begin{equation}
    x_{in} = (1-\delta)y + \delta q + \delta r x_{ni}
\end{equation}
The two equations above can be rewritten using $K_1$ and $K_2$ from~(\ref{Eq:coefficients-K1}) and  Eq.~(\ref{Eq:coefficients-K2}) respectively, 
\begin{equation}
    x_{in} = K_1 + K_2 x_{ni}
\end{equation}
\noindent The payoff of the mutant $\pi_n$ can be simplified using the linear relationship above in the following manner,
\begin{align}
   \pi_n &= \frac{1}{n-1} \sum_{i \neq n} bx_{in} - cx_{ni} \\
   &= \frac{1}{n-1} \sum_{i \neq n} b(K_1 + K_2x_{ni}) - cx_{ni} \\
   &= K_1b + (K_2b - c)x_{ni}
\end{align}
\end{proof}

\begin{proof}[Proof of Lemma~\ref{claim:algebraic-relationship}]
With some algebraic steps it can be confirmed that for $n>2$, $K_2b - c$ can be expressed as, $(\Lambda_1 r^2 + \Lambda_2 - \Lambda_3)/(1-B_3C_3)$. To complete the proof we only need to show that $1-B_3C_3 > 0$. From the expressions of $B_3$ and $C_3$ in Eq.~(\ref{Eq:coefficients-Bs}), it is easy to confirm that $0 \leq B_3 < 1$ and $0 \leq C_3 < 1$. Thus, $1-B_3C_3 > 0$. The signs of $K_2b - c$ and $\Lambda_1 r^2 + \Lambda_2 r - \Lambda_3$ are therefore always the same. 
\end{proof}

\begin{proof}[Proof of Theorem~\ref{Th:generic-Nash-equilibria}]
Resident strategies enforce a linear relationship on the mutant's payoff based on how cooperative the mutant is to the residents (Proposition~\ref{Prop:average-payoff-mutant-linear}). Residents can be classified into three groups: a) ones which enforce a zero slope on this relationship: $K_2b - c = 0$ (equalizer strategies) or b) ones which enforce a positive slope: $K_2b - c > 0$ or c) ones which enforce a negative slope: $K_2b - c < 0$. From Lemma~\ref{claim:algebraic-relationship} we know that it is sufficient to check the sign of the quadratic on $r$: $\Lambda_1 r^2 + \Lambda_2 r - \Lambda_3$ to group resident strategies into these three classes. Let's look at each case:
\\

\noindent \textit{When $\Lambda_1 r^2 + \Lambda_2 r - \Lambda_3=0$}: All resident strategies that equalize mutants' payoff are Nash equilibrium strategies. Since for all values of average cooperation the mutant's payoff is the same, every possible strategies is a best response to the resident. The resident strategy is therefore also a best response to itself. We classify these Nash strategies under \textit{generic Nash}.\\

\noindent \textit{When $\Lambda_1 r^2 + \Lambda_2 r - \Lambda_3 < 0$}: Resident strategies that enforce a negative slope on the mutant's payoff relationship necessitates full defection as the only best response for a mutant against the resident. Resident strategies that fully defect in a homogenous ($x = 0$) \textit{and} enforce a negative slope for the mutant player's payoff are thus best response to themselves. These are strategies that simultaneously satisfy, $\Lambda_1 r^2 + \Lambda_2 r - \Lambda_3 < 0$ and $y=q=0$ and additionally $\lambda = 0$ (if $\varepsilon \neq 0$). One special strategy that satisfies these properties is \textrm{ALLD}: $y=p=q=r=0$. For this strategy, $\Lambda_1 r^2 + \Lambda_2 r - \Lambda_3 = -\Lambda_3 < 0$. We also classify \textrm{ALLD} under the class of generic Nash.\\

\noindent \textit{When $\Lambda_1 r^2 + \Lambda_2 r - \Lambda_3 > 0$}: Resident strategies that enforce a positive slope on the mutant's payoff relationship necessitates full cooperation as the only best response for a mutant against the resident. Resident strategies that fully cooperate in a homogenous ($x = 1$) \textit{and} enforce a positive slope for the mutant player's payoff are thus best response to themselves. These are strategies that simultaneously satisfy, $\Lambda_1 r^2 + \Lambda_2 r - \Lambda_3 > 0$ and $y=p=1$ and additionally $\lambda = 0$ (if $\varepsilon \neq 0$). 
\end{proof}

\begin{proof}[Proof of Theorem~\ref{Th:two-player-Nash-equilibria}]
For $n=2$, the linear relationship between mutant's payoff and its average cooperation to the resident (now just a single player), $x_{ni}$ is given by,
\begin{equation}
    \pi_n = (1-\delta)y + \delta q + (\delta (p-q)b - c)x_{ni}
\end{equation}
When $p-q = c/(b\delta)$, the resident is an equalizer and therefore a Nash equilibrium (using the same argument as the previous proof). When $p-q < c/b$, the strategy that always defects against itself: $y=q=0$ is a Nash equilibrium. That is, $y=q=0$, $p < c/(b\delta)$ is also a Nash equilibrium. When $p-q < c/(b\delta)$ a strategy that always cooperates with itself: $y=p=1$ is a Nash equilibrium. That is, $y=p=1$ and $q < 1 - c/(b\delta)$ is  Nash equilibrium. 
\end{proof}

\begin{proof}[Proof of Proposition~\ref{Prop:cooperative-Nash}]
By definition, a generic cooperative Nash equilibrium are generic Nash strategies. So, strategies which result in $\Lambda_1 r^2 + \Lambda_2 r - \Lambda_3 \neq 0$ are ruled out. In addition, these strategies have cooperation in the homogenous population, $x \to 1$ as $\varepsilon \to 0$. In the limit of $\varepsilon \to 0$, \\
\begin{equation}
    \lim_{\varepsilon \to 0} x = \displaystyle \frac{(1-\delta)y + \delta q + \delta(n-2)q(\lambda + \gamma)}{(1-\delta p) + \delta q + \delta (n-2) q (\lambda + \gamma) + \delta(n-2)(\lambda + \gamma)(1-p)} 
\end{equation} \\
\noindent The right hand side of the above expression is 1 only if $y=p=1$. So, generic cooperative Nash strategies need to simultaneously satisfy $y=p=1$ and $\Lambda_1 (1-q)^2 + \Lambda_2 (1-q) - \Lambda_3 = 0$. 
\end{proof}
\begin{proof}[Proof of Proposition~\ref{Prop:GR-coop-Nash}]
A strategy of the form $(1,1,q^*_{\mathrm{G}}, 0, 1)$ already satisfies the first necessary condition for a generic cooperative Nash (from Proposition~\ref{Prop:cooperative-Nash}). It only needs to satisfy the second necessary condition, $\Lambda_1 (1-q^*_{\mathrm{G}})^2 + \Lambda_2 (1-q^*_{\mathrm{G}}) - \Lambda_3 = 0$ to be a cooperative Nash. For $\lambda = 0$ and $\gamma = 1$ (which is the case for this strategy), the quadratic on $(1-q^*_{\mathrm{G}})$ becomes only a linear in $(1-q^*_{\mathrm{G}})$ since $\Lambda_1 = b\delta^2(n-2) - b\delta^2(n-2) = 0$. The remaining coefficients $\Lambda_2, \Lambda_3$ are given by,
\begin{equation}
\begin{split}
   \Lambda_2 &= (1 + \delta(n-2))(b\delta + c\delta(n-2)) \\
   \Lambda_3 &= c(1 + \delta(n-2))^2
\end{split}
\end{equation}
Since $(1 + \delta(n-2)) > 0$, the solution of $q^*_{\mathrm{G}} = 1 - \Lambda_3/\Lambda_2$ is, 
\begin{equation}
    q^*_{\mathrm{G}} = 1 - \frac{c + c\delta(n-2)}{b\delta + b\delta(n-2)}
\end{equation}
\end{proof}
\begin{proof}[Proof of Proposition~\ref{Prop:Existence-Uniqueness-GR-equilibira} ]
A strategy of the form $(1,1,q,0,\gamma)$ is a generic cooperative Nash equilibrium if $\Lambda_1(1-q)^2 + \Lambda_2(1-q) - \Lambda_3 = 0$. For $\lambda = 0$, the coefficients of this quadratic are given by, 
\begin{equation}
\begin{split}
    \Lambda_1 &= -b\delta^2(1-\gamma)(1 + (n-2)\gamma)\\
    \Lambda_2 &= b\delta(1 + \delta(n-2)\gamma) + c\delta(1+\delta(n-2)\gamma)(1+(n-3)\gamma)\\
    \Lambda_3 &= c(1+\delta(n-2)\gamma)^2\\
\end{split}
\end{equation}
The coefficients have the following sign : $\Lambda_1 \leq 0$ $\Lambda_2 > 0$ and $\Lambda_3 > 0$. Therefore, the quadratic in $r:= 1-q$: $F(r) := \Lambda_1 r^2 + \Lambda_2 r - \Lambda_3$ has the following properties: 
\begin{enumerate}
    \item $F(0) < 0$
    \item $F'(0) > 0$
    \item $F''(r) < 0$
\end{enumerate}

\noindent Due to the above properties, if $\Lambda_1 + \Lambda_2 - \Lambda_3 = F(1) \geq 0$ then there is exactly one root of $r \in (0,1]$, which implies exactly one root for $q \in [0,1)$. This results in the following condition,

\begin{equation}
\begin{split}
\Lambda_1 + \Lambda_2 - \Lambda_3 = &(1 + \delta(n-2)\gamma)(b\delta + c\delta - c -c\delta\gamma) - b\delta^2(1-\gamma)(\gamma(n-2)+1) \geq 0 \\[10pt]
\implies &\delta(1-\gamma)((n-2)\gamma + 1) \leq (1 + \delta(n-2)\gamma)\left(1 + \frac{c}{b} - \frac{c}{b\delta} - \frac{c \gamma}{b} \right).
\end{split}
\label{Eq:show-F-equals-condition}
\end{equation}

\noindent We now show that, conversely, if there is only one root of $r \in (0,1]$ then it implies that $F(1) \geq 0$ (or equivalently, Eq. (\ref{Eq:show-F-equals-condition}) holds). In this case, since the quadratic $F(r)$ satisfies the above properties, a unique root $r^* \in (0,1]$  may occur in two ways: either $F'(r^*) > 0$ or $F'(r^*) = 0$. When $F'(r^*) > 0$ and $r^*$ is the unique root in $(0,1]$, it implies that $F(1) \geq 0$. We now show that the degenerate case $F'(r^*) = 0$ cannot arise. We prove this by contradiction.\\

\noindent We consider that $\exists r^* \in (0,1]$ such that $F(r^*) = F'(r^*) = 0$. Then, without loss of generality, one could express $F(r)$ as $F(r) = -u(r-r^*)^2$ where $u > 0$. By comparing coefficients we get

\begin{equation}
\begin{split}
-u &= \Lambda_1, \\
2ur^* &= \Lambda_2, \\
u(r^*)^2 &= \Lambda_3.
\end{split}
\end{equation}
\noindent This implies that $\Lambda_2^2 = -4\Lambda_1 \Lambda_3$. After some simplifications, we get
\begin{equation}
\begin{split}
    b\delta^2(1+\delta n \gamma)^2\left(1 + \frac{c}{b}(1-\gamma + n\gamma)\right)^2 &= 4bc\delta^2(1-\gamma)(1+n\gamma)(1+\delta n \gamma)^2, \\[10pt]
    \implies (2n'^2\theta^2 + 4n'\theta + 4\theta)\gamma^2 - 2n'\theta(1-\theta)\gamma + (1-\theta)^2 &= 0.
\end{split}
\end{equation}
\noindent Here $\theta = c/b < 1$ and $n' = n-1 \geq 2$. One can observe that the discriminant of this quadratic in $\gamma$, $\mathrm{D} := -4(1-\theta)^2(4n'^2\theta^2+4n'\theta + 4\theta) < 0$ thus implying that for fixed $n$ and $c/b<1$ there is no $\gamma$ and $r^*$ respectively in $[0,1]$ and $(0,1]$ such that $F(r)$ takes the form $-u(r-r^*)^2$. Therefore, a unique $r^* \in (0,1]$ such that $F(r^*) = 0$ requires $F'(r^*) > 0$ and thus $F(1) \geq 0$.

\end{proof}
\begin{proof}[Proof of Proposition~\ref{Prop:highest-generosity-among-coop-NE}]
The proof of this proposition follows from the result of the previous proposition. If for a value of $\gamma \in [0,1]$, the strategy $(1,1,q^*,0,\gamma)$ is the only generic cooperative Nash equilibrium strategy with $\lambda=0$, then it must be true that $F(1) \geq 0$ (the quadratic $F$ is defined in the previous proof). 
\\\\
\noindent As $F(1) \geq 0$ and $F(0) < 0$, there is exactly one root of the quadratic, $r^* \in (0,1]$. 
If $r^* < 1$, $F(r) > 0$, for all $r \in (r^*,1]$. Equivalently, $F(1-q) >0$ for all $q \in [0,q^*)$, where $q^* = 1 - r^*$. The strategy $(1,1,q,0,\gamma)$, for all $q \in [0,q^*)$, satisfies $y=p=1$ and $F(1-q) > 0$. So, by Theorem 1 (point 3), these strategies form a continuum of non-generic equilibria. Using the same Theorem, we can conclude that there is no Nash equilibrium of the form $(1,1,q,0,\gamma)$ for $q > q^*$. When $r^* = 1$, there are no non-generic Nash equilibria of this form because there is no $q$ for which $F(1-q) > 0$.   
\end{proof}
\bibliographystyle{unsrt}
\bibliography{bibliography.bib}